{
\documentclass[journal,10pt]{IEEEtran}

\usepackage{cite}

\usepackage{caption}
\captionsetup{font=scriptsize,labelfont=scriptsize}
\usepackage[cmex10]{amsmath}
\usepackage{stmaryrd}
\usepackage{algcompatible}
\usepackage{blkarray}
\usepackage{pgfplots}
\usepackage{blindtext}
\usepackage{graphicx}
\usepackage{mdwmath}
\usepackage{fixmath}
\usepackage{mdwtab}
\usepackage{bm}
\usepackage{amsmath}
\usepackage{amssymb}
\usepackage{multicol}
\usepackage{float}
\usepackage{bbm}
\usepackage{dsfont}
\usepackage{epstopdf}
\usepackage{cases}
\usepackage{color}
\usepackage{mathtools}
\usepackage{cleveref}
\usepackage[ruled,norelsize]{algorithm2e}
\usepackage{amssymb}
\usepackage{url}
\usepackage{pifont}
\usepackage{booktabs,makecell}
\usepackage{mathtools}
\usepackage[caption=false,labelformat=simple]{subfig}

\usepackage{makecell}

\makeatletter
\def\BState{\State\hskip-\ALG@thistlm}
\makeatother

\newenvironment{proof}[1][Proof]{\paragraph*{Proof}}{\ \rule{0.5em}{0.5em}}

\newtheorem{Proposition}{Proposition}
\newtheorem{lemma}{Lemma}

\newtheorem{myremark}{Remark}

\newcommand{\snr}{\mathrm{SNR}}
\newcommand{\sinr}{\mathrm{SINR}}

\newcommand{\tr}{\mathrm{tr}}

\newcommand{\Var}{\operatorname{Var}}
\newcommand{\Exp}{\mathbb{E}}

\newcommand{\Diag}{\operatorname{diag}}

\newcommand{\sinc}{\mathrm{sinc}}

\makeatletter
\newcommand{\removelatexerror}{\let\@latex@error\@gobble}
\makeatother

\begin{document}
\title{
Massive MIMO Channels with Inter-User Angle Correlation: Open-Access Dataset, Analysis and Measurement-Based Validation
}
\author{
\IEEEauthorblockN{Xu Du and Ashutosh Sabharwal}
\thanks{
The authors are with the Department of Electrical and Computer Engineering, Rice University, Houston, TX, 77005 (e-mails: xdurice@gmail.com, ashu@rice.edu).
This work was partially supported by NSF grants 1827940, 2016727 and 2120363, and a grant from Qualcomm, Inc.
Xu Du presented part of this paper as a section in his thesis defense at Rice University in~\cite{du2019phd_thesis}.
}
}

\bibliographystyle{libs/IEEEbib}
\maketitle
\IEEEpeerreviewmaketitle
\begin{abstract}
In practical propagation environments, different massive MIMO users can have correlated angles in spatial paths.
In this paper, we study the effect of angle correlation on inter-user channel correlation via a combination of measurement and analysis.
We show three key results.
First, we collect a massive MIMO channel dataset for examining the inter-user channel correlation in a real-world propagation environment; the dataset is now open-access.
We observed channel correlation higher than $0.48$ for {\sl all} close-by users.
Additionally, over $30$\% of far-away users, even when they are tens of wavelengths apart, have inter-user channel correlation that is at least twice higher than the correlation in the i.i.d.\ Rayleigh fading channel.
Second, we compute the inter-user channel correlation in closed-form as a function of inter-user angle correlation, the number of base-station antennas, and base-station inter-antenna spacing.
Our analysis shows that inter-user angle correlation increases the inter-user channel correlation.
Inter-user channel correlation reduces with a larger base-station array aperture, i.e., more antennas and larger inter-antenna spacing.
Third, we explain the measurements with numerical experiments to show that inter-user angle correlation can result in significant inter-user channel correlation in practical massive MIMO channels.
\end{abstract}

\section{Introduction}~\label{sec:intro}
Massive multiple-input multiple-output (MIMO) is one of the key $5$G components due to its potential benefits, including increased spectral efficiency~\cite{ngo2013energy, yang2013performance}, wider coverage region~\cite{marzetta2016fundamentals,bjornson2017massive,larsson2014massive}, decreased channel variation~\cite{marzetta2016fundamentals,bjornson2017massive}, reduced interference~\cite{marzetta2016fundamentals,bjornson2017massive}, improved reliability~\cite{du2019balance}, and simplified physical layer precoding~\cite{ngo2013energy}.
A commonly adopted model of massive MIMO, often labeled as ``favorable propagation''~\cite{marzetta2010noncooperative, 6375940, ngo2013energy, favorable_prop}, assumes that  the user channels are \mbox{(near-)orthogonal} to each other when the base-station is equipped with more than tens of antennas. This model follows if we assume that the channel between each transmit-receive antenna pair is generated by an independent and identically distributed (i.i.d.) process.

The inter-user channel orthogonality is a critical assumption behind many desirable properties of massive MIMO.
And inter-user channel correlation can reduce throughput~\cite{6375940,matthaiou2018does, marzetta2016fundamentals, bjornson2017massive}, increase interference~\cite{pratschner2019does, marzetta2016fundamentals, bjornson2017massive}, lower reliability~\cite{6375940, marzetta2016fundamentals, du2019balance}, and raise precoding complexity~\cite{ngo2013energy, yang2013performance, marzetta2016fundamentals, bjornson2017massive}.
For the i.i.d.\ Rayleigh fading channel, it is well-known that the user channels become orthogonal to each other as the number of base-station antennas increases. The inter-user channel correlation reduces at rate $1/\sqrt{M}$, where $M$ is the number of base-station antennas.
For line-of-sight (LOS) channels with various array configurations, past work~\cite{favorable_prop,wu2017favorable,matthaiou2018does} demonstrates that the inter-user channel correlation reduces at rate $O\left(1/\sqrt{M}\right)$ if the spatial paths have statistically independent angles.
Therefore, if the user channels are statistically independent or consist of paths with statistically independent angles, user channels should be near-orthogonal for practical massive MIMO systems with more than tens of base-station antennas.
As we move away from the independently distributed user channel assumption, relatively little is known about the massive MIMO inter-user channel correlation.

\begin{figure}[htbp]
  \centering
  \includegraphics[width=0.48\textwidth]{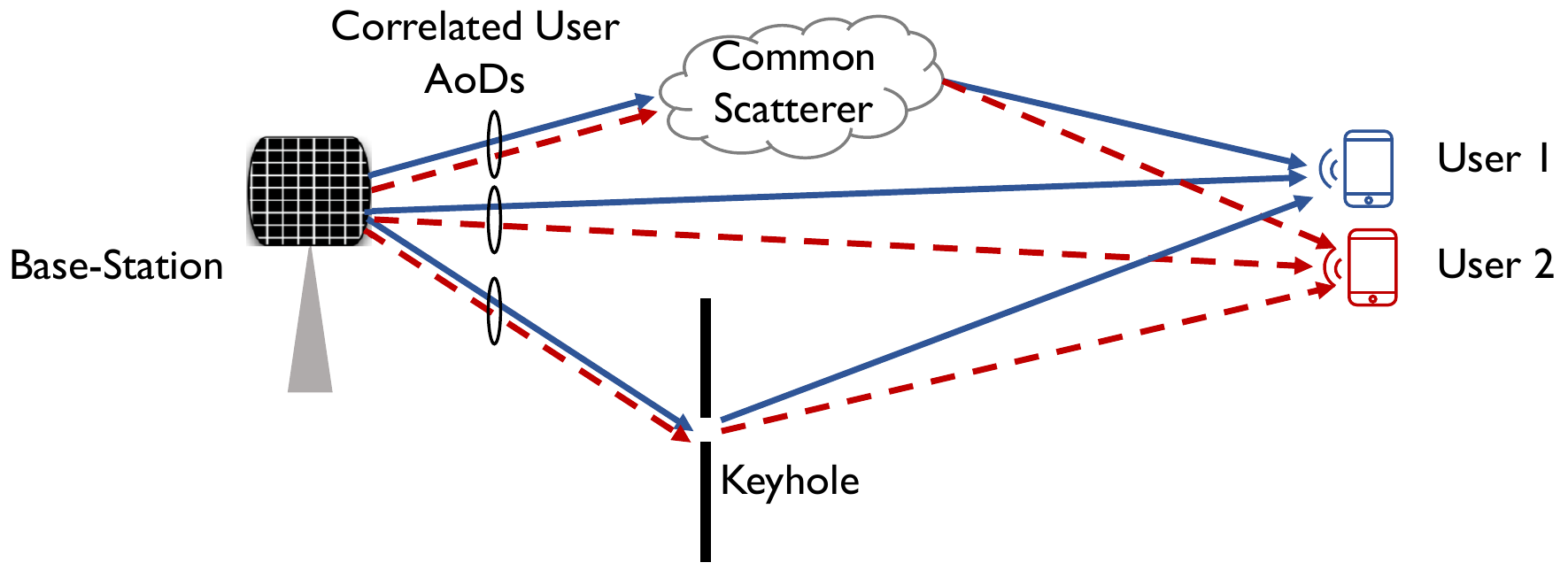}
  \caption{
  The figure illustrates a two-user massive MIMO downlink channel consisting of three path pairs with inter-user angle correlation.
   Two users are in geo-location proximity to each other.
   The solid blue (dashed red) lines represent spatial paths from the base-station towards User $1$ (User $2$).
   In this figure, the inter-user angles are close for three reasons: common scatterers, geo-location proximity, and keyhole effect.
   Section~\ref{sec:RayAnalysis} will prove that such inter-user angles-of-departure correlation can lead to significant inter-user channel correlation in practical massive MIMO systems.
   Note that the angles-of-arrival at the users can be distinct. And this paper does not assume correlation among the angles-of-arrival at the users.
  }~\label{fig:cha_mdl}
\end{figure}

However, in practical propagation environments, paths of different users could have correlated angles.
Fig.~\ref{fig:cha_mdl} depicts an example of a two-user channel consists of three paths to each user.
In Fig.~\ref{fig:cha_mdl}, the spatial path pairs have correlated angles due to three reasons.
The first reason for inter-user angle correlation is sharing common scatterers, where two-user paths share almost the same Angle-of-Departure (AoD) from the base-station.
The second pair is LOS paths, which have correlated angles because the users are close by to each other.
Finally, the third pair paths have correlated angles due to the spatial keyhole~\cite[Chapter 7.3]{marzetta2016fundamentals}, which is pictured by a slot through a wall. The keyholes can lead to close-by angles for even far-away users.
In summary, the inter-user angle correlation can exist in practical propagation environments.


The above examples motivate
the question \emph{how does inter-user angle correlation impact the inter-user channel correlation?}
Towards that end, our main contributions are as follows:
   \begin{enumerate}
     \item We collect a diverse massive MIMO channel dataset with a $64$-antenna planar array.
     Our evaluation includes line-of-sight (LOS) and non-line-of-sight (NLOS) propagation environments with more than $11,500$ unique user channel vectors from $225$ different locations.
     The measured channel dataset is publicly available  at~\cite{open_data_link}.
     With a $64$-antenna base-station array, the channel correlation between {\sl all} close-by users is higher than $0.48$.
     In our dataset, close-by users' channel correlation reduces very slowly when the number of base-station antennas is larger than $36$.
     For far-away users, $30.56$\% user pairs have channel correlation that is at least twice higher than the correlation in the i.i.d.\ Rayleigh fading channel.
     The observed user correlation is near-constant across the measured $20$ MHz Orthogonal Frequency-Division Multiplexing (OFDM) band.

     \item      We compute the inter-user channel correlation for massive MIMO with inter-user angle correlation by using a spatial channel model.
     We characterize the channel correlation in closed-form to explain the measurement findings.
 In the special case, when at least one inter-user path pair always share an exact angle, the inter-user channel correlation converges to a positive constant that does not reduce with the number of base-station antennas $M$.
     For space-constrained systems where the base-station inter-antenna spacing decreases inversely proportionally with $M$, the inter-user channel correlation also converges to a positive constant as $M$ increases to infinity.
     In all other cases with fixed base-station inter-antenna spacing, the inter-user channel correlation will converge to zero
     as $M\to \infty$.
     The inter-user channel correlation increases with the inter-user angle correlation, and reduces with the base-station inter-antenna spacing.
     Therefore, increasing the base-station array aperture with more antennas or larger inter-antenna spacing can reduce the inter-user angle correlation induced inter-user channel correlation.

     \item
      We use numerical experiments to explain the measurement and verify the inter-user correlation analysis in LOS and NLOS environments.
       The simulation results demonstrate the analysis accuracy in both finite-array and large-array regimes.
     Numerical experiments confirm that inter-user angle correlation increases inter-user channel correlation.
The simulations further verify our theoretical prediction that increasing the base-station array aperture can offset the inter-user channel correlation induced by inter-user angle correlation.
    \end{enumerate}

 \noindent
{\sl Related Work}:
One of the early measurements~\cite{Bell_lab_LOS_measurement} demonstrated that the {\sl average} inter-user channel correlation of far-away users reduces as the number of base-station antennas increases.
By measuring tens of close-by user channels with uniform linear array~\cite{Gao-mea,chen2017scaling,pratschner2020measured}, uniform cylindrical array~\cite{Gao-mea,Lund_MU2015}, uniform rectangular planar array\cite{demark_mea,7869082}, recent channel measurements find that the channels of close-by users might have high correlation.

Therefore, the inter-user channel correlation had been observed to behave differently for users with different physical distances.
During our measurements, we observed high channel correlation not only for {\sl all} close-by users, but also for $30.56$\% far-away user pairs.
Using spatial signal processing, we find that the inter-user angle proximity is a likely cause of high inter-user channel correlation for both close-by and far-away users.

Past  work~\cite{marzetta2016fundamentals, favorable_prop, wu2017favorable, yang2017massive}  analyzed  the  inter-user channel  correlation  in  the  LOS  environment  without  inter-user  angle  correlation  for  a  base-station  with  the  uniform linear array.
Recent work~\cite{masouros2015space, pratschner2019does} extended the past work to consider the effect of limited base-station linear array aperture.
Section~\ref{sec:RayAnalysis}  presents  a  generalized  inter-user channel  correlation analysis  for  a  multi-carrier  OFDM  system  in  a multi-path environment with inter-user channel correlation, and the uniform  planar  array.
Compared  to~\cite{marzetta2016fundamentals, favorable_prop, wu2017favorable, masouros2015space, yang2017massive}, Section~\ref{sec:RayAnalysis} captures the effect of the inter-user angle correlation, multi-path, cross subcarrier channel correlation variation, and the two-dimension array.
The generalization enables the explanation of the measurement findings.
And the existing analysis can be viewed as specialized versions of Section~\ref{sec:RayAnalysis}.

There is also prior research~\cite{kolmonen2010measurement, liu2012cost, QuaDRiGa, NYUSIM, pratschner2020measured, spatially_consistence} on modeling the time-variant channels along the trajectory of a single mobile user by considering the spatial scattering clusters.
In each scattering cluster, the simulator~\cite{kolmonen2010measurement, liu2012cost, QuaDRiGa, NYUSIM, pratschner2020measured, spatially_consistence} randomly generates multi-path components according to specific large-scale parameters.
This research stream is usually referred to as the “spatially consistent channel” and is justified by fitting to channel measurements.
For example,  QuaDRiGa~\cite{QuaDRiGa}  adopted the simulation-based approach to model the time-variant channel of a  MIMO downlink mobile user.
And NYUSIM~\cite{NYUSIM} modeled the mm-wave bands that are up to $140$ GHz.  Reference  inside\cite{pratschner2020measured, spatially_consistence, QuaDRiGa}  provides  more  details  about  this  research  stream.
As this paper focuses on the inter-user channel correlation measurement and analysis, Section~\ref{sec:RayAnalysis} leverages a spatial channel model to explain the over-the-air findings.
We consider it a future direction to improve the analysis in Section IV by adopting a different channel model, like the spatially consistent channel models.

Finally, there exists massive MIMO work that suppresses the intra-cell and inter-cell interference in the spatial domain by optimized beamforming design~\cite{bjornson2017massive, marzetta2016fundamentals,adhikary2013joint, ngo2013energy, yang2013performance,6375940}, improved channel estimation quality~\cite{bjornson2017massive, marzetta2016fundamentals,6375940}, and interference-aware user grouping~\cite{bjornson2017massive, marzetta2016fundamentals,adhikary2013joint,7588195,6375940}.
This paper focuses on the inter-user channel correlation measurement and analysis.
And the impact on massive MIMO system design is a crucial direction beyond the scope of this paper.

The paper is organized as follows. Section~\ref{sec:SysMdl} describes the system model.
We next present a channel measurement campaign and characterize the inter-user channel correlation over-the-air in Section~\ref{sec:measurement}.
Section~\ref{sec:RayAnalysis} explains the measurement findings by computing the inter-user channel correlation for systems with inter-user angle correlation.
Section~\ref{sec:Numerical} uses numerical experiments to verify the analytical results in Section~\ref{sec:RayAnalysis}.
Finally, Section~\ref{sec:Conclude} concludes the main findings of this paper.

{\sl Notations}:
Boldface represents vectors and matrices.
We use $|\cdot|$ to denote the magnitude of a complex number.
And the $l_2$-norm of a complex vector is $\| \cdot \|$.
The space of the real number is $\mathbb{R}$ and $\mathbb{C}$ is the complex space.
The complex Gaussian distribution is denoted by $\mathcal{CN}$.
The conjugate transpose and the trace of a matrix $\mathbf{H}$ are denoted as $\mathbf{H}^{H}$ and $\tr \mathbf{H}$, respectively.
The diagonal matrix is $\Diag\left(\cdot\right)$, and an $M$-element all-ones vector is $\mathbf{1}_{M}$.
For a random variable $X$, the expected value and variance are denoted by $\Exp\left[X\right]$ and $\Var\left[X\right]$, respectively.

\section{System Model}\label{sec:SysMdl}
\subsection{System Setup}\label{subsec:sys_setup}
We consider a single-cell massive MIMO OFDM downlink system with an $M$-antenna base-station and $K$ single-antenna mobile users.
The channel consists of $N$ OFDM subcarriers.
We assume that the base-station is time synchronized with the $K$ users.
For subcarrier $n$, the downlink channel is $\mathbf{G} \mathbf{H}_{n}$, where $\mathbf{G}=\Diag\left(\sqrt{\gamma_1}, \sqrt{\gamma_2}, \dots, \sqrt{\gamma_K}\right)$ is the large-scale channel gains and $\mathbf{H}_{n}\in\mathbb{C}^{K \times M}$ denotes the small-scale fading for Subcarrier $n$.
The downlink channel of User $k$, $\mathbf{h}_{k, n}\in\mathbb{C}^{M}$, is the $k$-th row of the small fading channel matrix, whose norm satisfies $\left\|\mathbf{h}_{k, n}\right\| =\sqrt{M}$.
The base-station adopts downlink beamforming to serve all $K$ users at the same time. Let the downlink beamforming matrix be $\mathbf{V}_{n}\in\mathbb{C}^{M \times K}$. The received signals by users then equal
\begin{equation}
\mathbf{y}_{n} = \mathbf{G} \mathbf{H}_{n}\mathbf{V}_{n}\mathbf{s}_{n} + \mathbf{w}_{n}, \quad n=1,\dots,N.~\label{equ:sig_mdl}
\end{equation}
The $\mathbf{w}_{n}\in \mathbb{C}^{K}$ is the additive noise, whose elements follow independent standard complex Gaussian distribution $\mathcal{CN}\left(0, 1\right)$.
Each symbol in the signal sequences $\mathbf{s}$ follows $\mathcal{CN}\left(0, 1\right)$.
And the total transmission power is $P$, which equals $\Exp\left[\tr \mathbf{V}_{n}\mathbf{V}_{n}^{H}\right]$.

Prior to the data transmission, the base-station estimates the downlink channel $\hat{\mathbf{H}}_{n}\in\mathbb{C}^{K \times M}$ for beamforming.
A popular massive MIMO channel estimation model is~\cite{lottici2002channel,marzetta2016fundamentals}
\begin{equation}
\mathbf{H}_{n} = \hat{\mathbf{H}}_{n} + \tilde{\mathbf{H}}_{n}, ~\label{equ:h_hat}
\end{equation}
where $\tilde{\mathbf{H}}_{n} \in \mathbb{C}^{K \times M}$ is the channel estimation error matrix.
Denote $\hat{\mathbf{h}}_{k,n}$, $\tilde{\mathbf{h}}_{k,n}$, and $\mathbf{v}_{k,n}$ as the estimated channel, estimation error, and beamforming vector of User $k$ in Subcarrier $n$, respectively.
The received signal of User $k$ equals~\cite{marzetta2016fundamentals,bjornson2017massive}
\begin{align}
y_{k,n} = & \sqrt{\gamma_{k}} \hat{\mathbf{h}}_{k,n}^{H} \mathbf{v}_{k,n} s_{k,n}
 + \sum_{j\neq k}\sqrt{\gamma_{k}} \hat{\mathbf{h}}_{k,n}^{H} \mathbf{v}_{j,n} s_{j,n} \notag
 \\
 & + w_{k,n}
 - \sqrt{\gamma_{k}} \tilde{\mathbf{h}}_{k,n}^{H} \mathbf{V}_{n} \mathbf{s}_{n},\label{equ:y_k}
\end{align}
where $w_{k,n}$ is the additive noise, and $s_{k,n}$ is the transmitted symbol.
The first term denotes the desired signal. And the rest three terms in~\eqref{equ:y_k} represent interference from the beamforming algorithm, noise, and interference from the imperfect channel knowledge, respectively.

The interference terms in~\eqref{equ:y_k} imply that inter-user channel correlation reduces effective $\sinr$ due to increased interference.
For example, when the base-station adopts the conjugate beamforming~\cite{marzetta2016fundamentals,bjornson2017massive}
based upon perfect channel knowledge, $ \mathbf{v}_{k,n} = \mathbf{h}_{k,n}^{H}$.
The interference power caused by inter-user channel correlation then equals
$
\gamma_{k} \sum_{j\neq k} \left|\mathbf{h}^{H}_{j,n} \mathbf{h}_{k,n}\right|^2,
$
which is determined by the inter-user channel dot products.
Therefore, the inter-user channel correlation reduces the effective $\sinr$, hence the achievable transmission rate.
This paper focuses on characterizing the inter-user channel correlation with measured massive MIMO channels, and providing theoretical explanations.
Textbooks~\cite{marzetta2016fundamentals,bjornson2017massive} and the included references provide an comprehensive review of massive MIMO system performance with more complex settings.

\subsection{User Correlation Definition}
Inspired by the past massive MIMO literature~\cite{marzetta2010noncooperative, 6375940, larsson2014massive, ngo2013energy, yang2013performance},
this paper captures the inter-user channel correlation by the cosine similarity.
This paper only calculates the inter-user channel correlation for channels on the same subcarrier.
To simplify the notation, we will omit the subcarrier index $n$ in $\mathbf{h}_{k,n}$ when the subcarrier index is clear from the context.
The correlation of User $k$ and User $k^{'}$ is defined as
\begin{equation}
 \sqrt{\Var \left[\frac{\mathbf{h}_{k}^{H}\mathbf{h}_{k^{'}}}{\left|\mathbf{h}_{k}\right| \left|\mathbf{h}_{k^{'}}\right|} \right] }= \frac{1}{M} \sqrt{\Var \left[\mathbf{h}_{k}^{H}\mathbf{h}_{k^{'}}\right]}, \label{equ:alpha_def}
\end{equation}
where the variance operator is over the joint distribution of user channels. The desired pairwise channel orthogonality means
\begin{equation}
\frac{1}{M} \sqrt{\Var \left[\mathbf{h}_{k}^{H}\mathbf{h}_{k^{'}}\right]} = 0, \ \text{for all} \ k\neq k^{'} \ \text{and} \ k,k^{'}=1,2,\dots, K. ~\label{equ:favorable_condition}
\end{equation}
Note that this paper does not choose the mean operator $\Exp\left[\mathbf{h}_{k}^{H}\mathbf{h}_{k^{'}}\right]$ since it always equals to zero when the phases of spatial paths follow independent uniform distribution over $\left[0, 2\pi\right]$.
Please see~\cite[7.1]{marzetta2016fundamentals} and~\eqref{equ:main_finite_mean} in Appendix~\ref{appendix:main_finite} for examples in single-path and multi-path environments.

To understand the deviation of practical inter-user channel correlation from~\eqref{equ:favorable_condition}, this paper presents a combination of measurement campaign and theoretical analysis.
Section~\ref{sec:measurement} uses over-the-air measurements to characterize inter-user channel correlation with a base-station that has $64$ antennas.
Section~\ref{sec:RayAnalysis} computes the inter-user channel correlation with inter-user angle correlation in closed-form.
The theoretical results are numerically verified in Section~\ref{sec:Numerical}.

Following the past  MIMO inter-user channel correlation measurements~\cite{Gao-mea,chen2017scaling,pratschner2020measured} and analysis~\cite{favorable_prop, wu2017favorable,masouros2015space, pratschner2019does}, this paper uses the i.i.d. Rayleigh fading channel as the baseline benchmark.
This baseline channel model can be justified with the central limit theorem via considering rich scattering environments with an infinite number of paths from independent scatters~\cite{marzetta2016fundamentals}.

\section{Measured Inter-User Channel Correlations}\label{sec:measurement}
In this section, we examine inter-user channel correlation based on Over-The-Air (OTA) measured channels.
The collected channel data set is also open-accessed for download at~\cite{open_data_link}.
We remark that~\cite{open_data_link} also contains measured channel dataset for full-duplex massive MIMO~\cite{SoftNull} and FDD massive MIMO~\cite{zhang2018directional}.

\subsection{Measurement Setup}~\label{subsec:measurement_detail}
\begin{figure*}[htbp]
\centering
  \begin{minipage}{.222\linewidth}
  \centering
  \subfloat[Argos Platform]{\label{fig:array}
                            \includegraphics[height=1\linewidth]{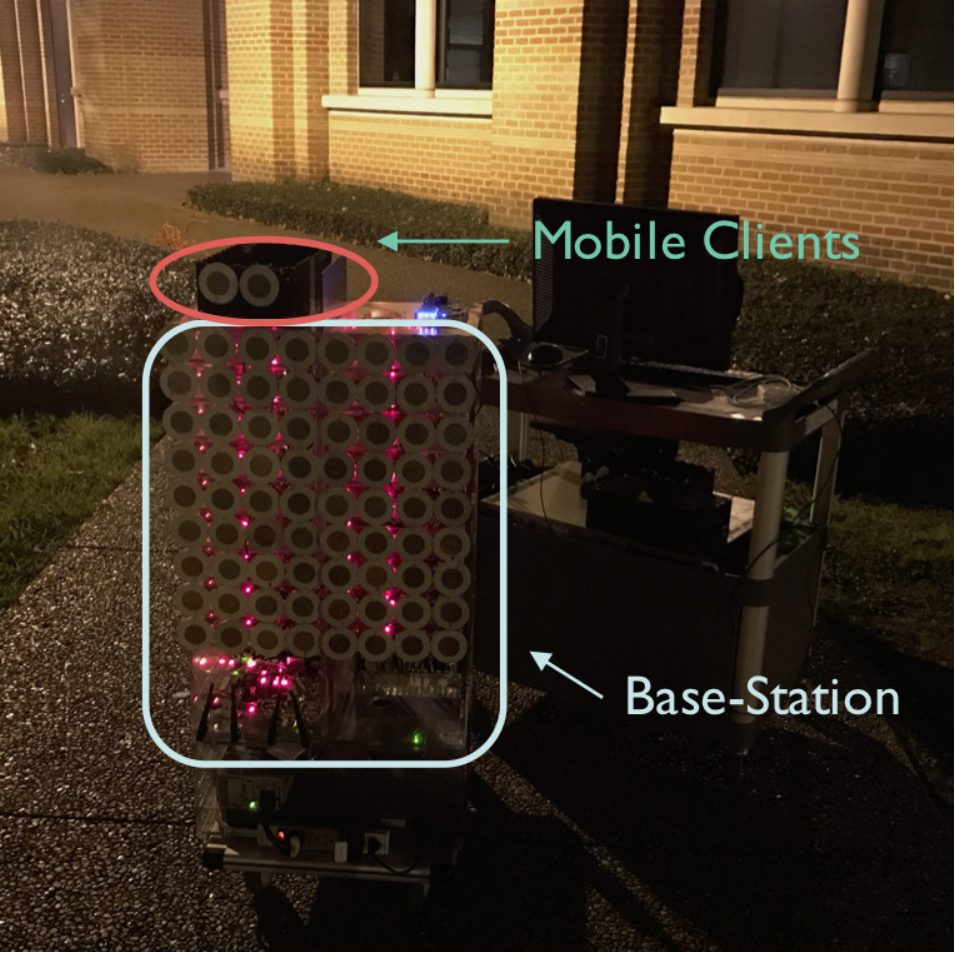}}
  \end{minipage}
  \begin{minipage}{.444\linewidth}
  \centering
  \subfloat[Measurement Locations]{\label{fig:exp-locs}
                            \includegraphics[height=0.5\linewidth]{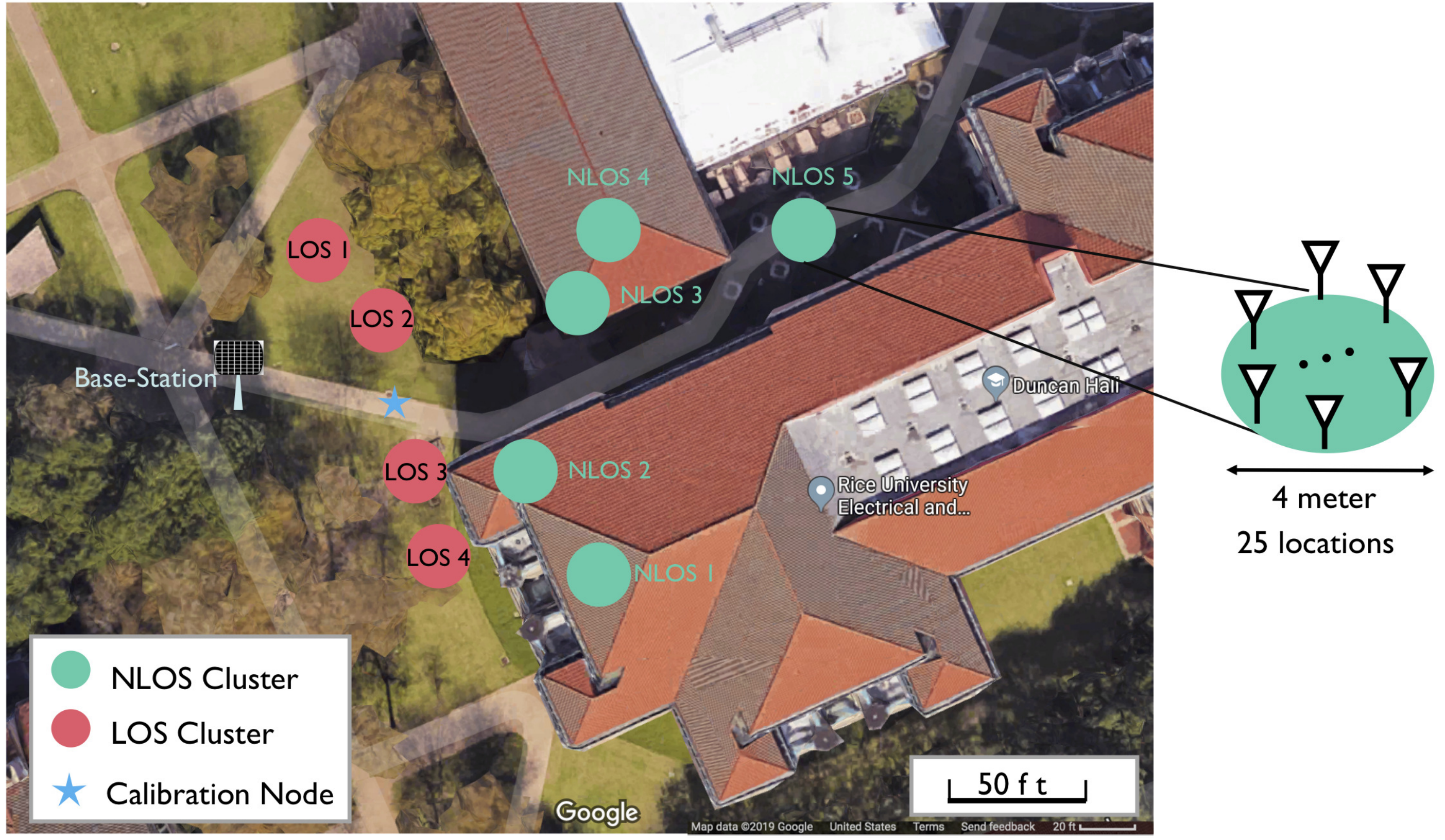}}
  \end{minipage}

  \begin{minipage}{.222\linewidth}
      \centering
  \subfloat[Measurement Setup]{\label{fig:exp_setup}
                            \includegraphics[height=0.8\linewidth]{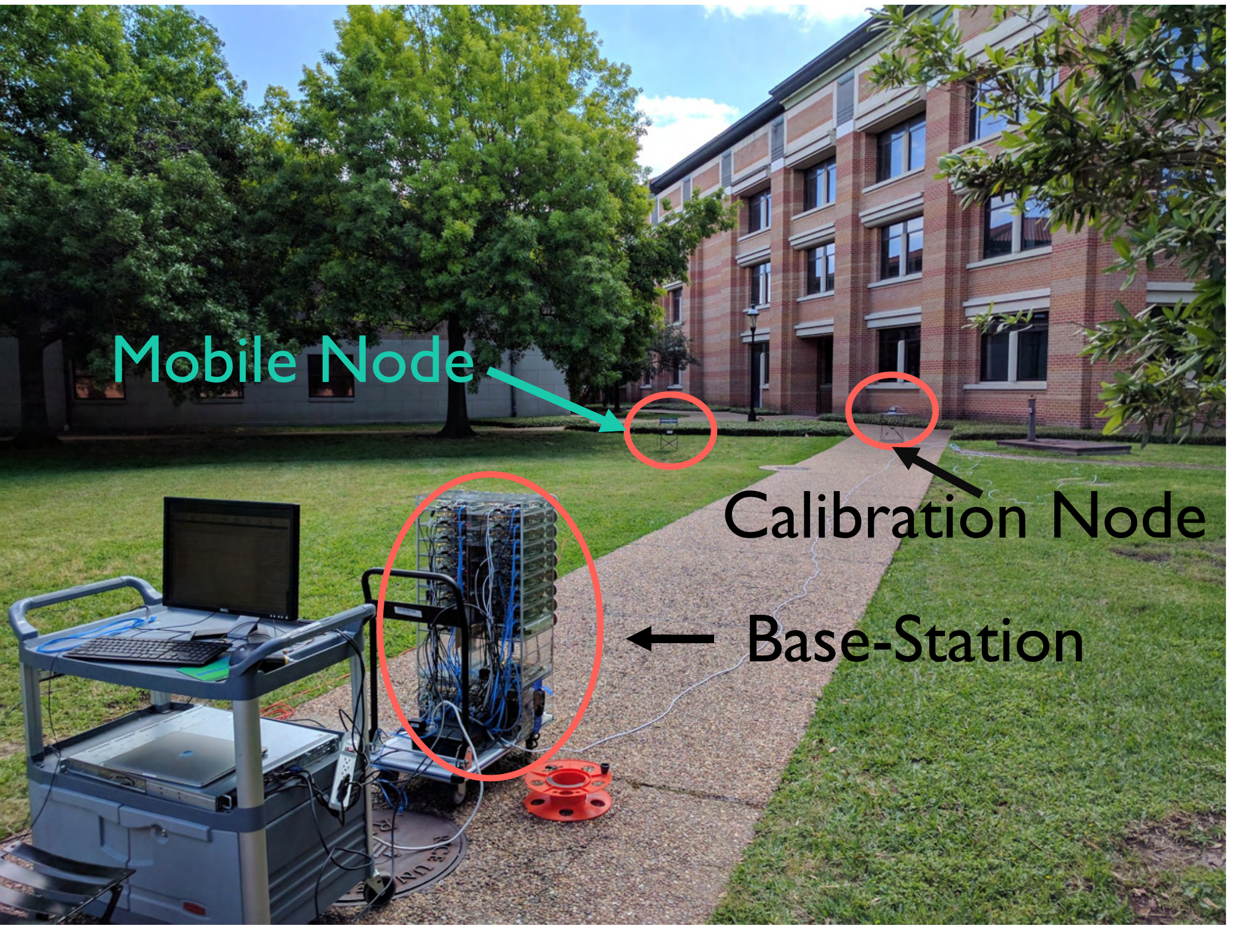}}
  \end{minipage}
  \begin{minipage}{.444\linewidth}
  \centering
  \subfloat[Calibration Node Setup]{\label{fig:calib}
                            \includegraphics[height=0.4\linewidth]{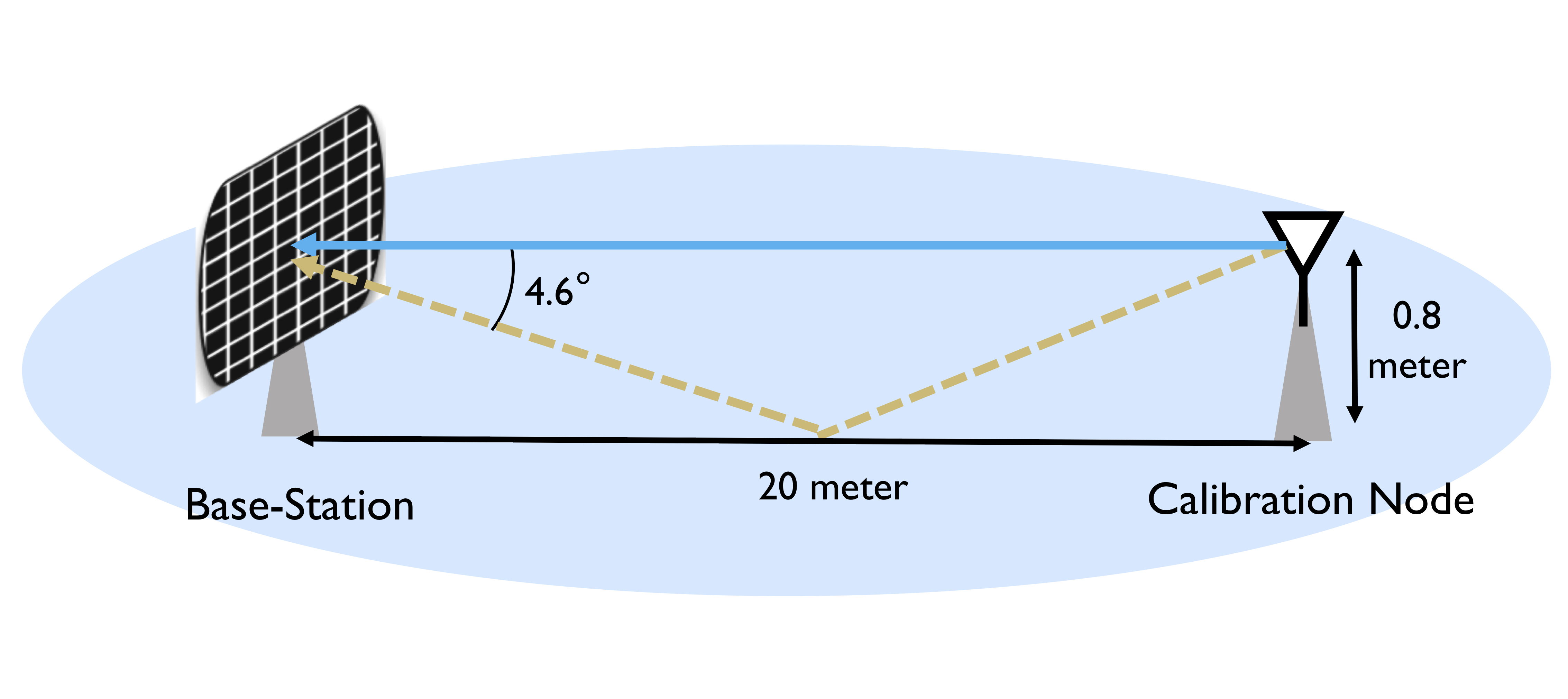}}
  \end{minipage}
\caption{Argos~\cite{shepard2012argos} Massive MIMO base-station and the over-the-air measurements setup.
The background map of Fig.~\ref{fig:exp-locs} is from Google Maps~\cite{google_map}.
}
\end{figure*}

\subsubsection{Measurement Platform}
  We used the Argos platform to measure the channels from mobile nodes to a $64$-antenna massive MIMO base-station.
  Both the mobile nodes and the base-station are developed based upon the WARP software defined radio modules~\cite{warp}.
  Each of the used WARP module is equipped with four active radio chains, a clock board, and a Field-Programmable Gate Array (FPGA).
  The FPGA and radio chains are programmed via extending the WARPLab~\cite{warp}.
  All WARP modules are controlled by a Argos central controller via a customized Argos Hub.
  The Argos Hub is responsible for data switching, WARP module synchronization, and clock distribution.
  A dedicated design and review of the Argos platform and WARP platform are presented in~\cite{shepard2012argos} and~\cite{warp}, respectively.
  And Section 4 of~\cite{shepard2012argos} provides the implementation details, which includes the function blocks at the mobile users, base-station, and Argos controller.

  During the measurement, the Argos platform measures the uplink channel by letting the mobile nodes send 802.11 frames that contain the channel estimation pilots in the Physical Layer Convergence Protocol (PLCP) preamble, whose training frame structure can be found in~\cite[Fig. 18-4]{802-11-2012}.
  For Argos, the transmitters transmit pilots that are located in the Long Training Sequence (LTS) in the 802.11 PLCP preamble.
  The base-station then measures and records the raw IQ samples.
  Finally, the channel is estimated by correlating the measured IQ samples using the pilot sequence with Fast Fourier Transform~\cite{warp}.

 With the Argos platform, we measure the OTA channels on the campus of Rice University near the Anne and Charles Duncan Hall (N $29.720138^{o}$, W $95.39876^{o}$).
The users are in a near-stable environment.
The measurements are on the $2.4$ GHz Wi-Fi ISM band with a bandwidth of $20$ MHz.
Both the base-station and users are equipped with patch antennas with $3$-dB beam-width of around $120$-degree. The $8 \times 8$ $64$-antenna array and a mobile user are shown in Figure~\ref{fig:array}.

\subsubsection{Measured Locations}
The channels between the base-station antennas and users from $4$ line-of-sight (LOS) clusters and $5$ non-line-of-sight (NLOS) clusters are measured.
Each cluster is of circle-shape with a diameter of around $4$ meters.
Figure~\ref{fig:exp-locs} presents the locations of the massive MIMO base-station and user clusters.
In each cluster, we measure the user channels of $25$ {\sl different} uniformly selected locations.
Therefore, the minimum distance between two users from the same cluster is about $0.7$ meters.
In total, we measure the mobile user channels in $225$ locations.
For each location, we measure the uplink channel between a mobile user and the base-station of the $52$ data-carrying subcarriers over the $20$ MHz band in more than $140$ frames.
In our evaluations, we will utilize all subcarriers from the first frame of each location.
In total, we characterize the inter-user channel correlation over-the-air based on more than $11000$ unique channel vectors of length $64$.

\subsubsection{Array Calibration for Antenna Phase Mismatch}
For the Argos~\cite{shepard2012argos} massive MIMO base-station, each antenna is connected to an individual radio-frequency chain.
During channel measurements, each base-station antenna estimates its channel coefficient by correlates the received signal with an uplink pilot sequence.
As the phase-locked loop component of each radio-frequency chain introduces an independent random phase to each antenna, the collected channel contains mismatched phases across different base-station antennas and different measured locations.
Such phase mismatch raises two critical issues in the examination of inter-user channel correlation.
Firstly, the phase mismatch across base-station antennas makes the (uplink) angle-of-arrival estimation infeasible.
The reason is that the angle-of-arrival estimation requires a phase synchronized array~\cite{MUSIC-review}.
Secondly, phase mismatches across different measurements obstruct the inter-user channel correlation computation of channels measured at different locations.
The challenges come from that the correlation computation~\eqref{equ:alpha_def} requires that each base-station antenna have a fixed (relative) phase offset during the measurements of different users.
Thus, array phase calibration is essential for both the angle estimation and the inter-user channel correlation computation.

To calibrate the phase mismatch, we introduce a calibration node in all measurements.
The calibration node is placed perpendicular to the array plane at the direction of elevation angle $0^{\circ}$ and azimuth $0^{\circ}$.
The distance between the calibration node and the base-station array is about $20$ meters.
The main path is at $\left(0^{\circ},0^{\circ}\right)$.
Fig.~\ref{fig:calib} shows that there might exist an additional ground reflection path, which is about $4.6^{\circ}$ elevation away from the main path.
Since the ground reflection path should be of much lower gain and the angle difference is small, the channel between the calibration node and the array can be modeled as $\mathbf{1}_{M}$.
Therefore, by measuring the uplink channel from the calibration node to the base-station, we obtain the relative phase mismatch of each base-station antenna.
To determine the calibrated channel of each location, we measure the channels of the mobile and the calibration nodes in {\sl all} channel measurements.
The calibrated channel of each location is obtained by the element-wise division of the mobile node channel by the calibration node channel.

\subsection{Experimental Inter-User Channel Correlation Examination}~\label{subsec:ota_findings}
This subsection presents the measurement findings on inter-user channel correlation in the LOS and NLOS propagation environments.
We are primarily interested in how often high inter-user channel correlation happens for nearby users and users that are tens of meters apart.
And the experimental characterization of inter-user channel correlation as a function of user geo-position separation is a future research problem.

During measurements, we measure the channel between the mobile user and all $64$ base-station antennas.
In evaluations, when computing channel correlation for $M<64$, we sub-sample square sub-arrays of dimension $\sqrt{M} \times \sqrt{M}$.
For example, when $M=36$, all antennas from the first $6$ rows and first $6$ columns are selected.
The base-station array indexing and our implementation of the array sub-sampling can be found in~\cite{open_data_link}.
As measured channel $\snr$ of each antenna is over $15$ dB in all channel measurements, we treat the measured channels as perfect.

\noindent
\emph{OTA Finding 1 -- Close-by user channels correlation decays very slowly when $M>36$}:
We first evaluate the inter-user channel correlation between close-by users from the same cluster, i.e., {\sl intra-cluster correlation}.
Fig.~\ref{fig:intra_cor} presents the computed intra-cluster correlation of the $9$ clusters whose locations are described in Fig.~\ref{fig:exp-locs}.
Each data point on each curve is obtained by averaging over $15600$ unique channel pairs among the $25$ locations in all $52$ subcarriers.
And the error bars represent the variation of intra-cluster correlation across the $52$ different subcarriers.
The inter-user channel correlation in the i.i.d. Rayleigh fading channel is also presented as a benchmark.

For user channels measured in the LOS clusters, intra-cluster correlation reduces very slowly as $M$ increases.
Over the measurements, we observed intra-cluster correlation over $0.8$ for {\sl all} LOS clusters even when the base-station is equipped with $64$ antennas.
For the NLOS clusters, intra-cluster correlation reduces quickly when $M$ is smaller than $25$. When $M>36$, the intra-cluster correlation of the NLOS clusters reduces very slowly or does not decrease.
Furthermore, the NLOS clusters have lower intra-cluster channel correlation than the LOS clusters.
The lower intra-cluster correlation in NLOS clusters will later be confirmed by the theoretical results in Section~\ref{sec:RayAnalysis}.
In summary, for {\sl all} clusters, an intra-cluster channel correlation of more than $0.48$ is observed with the $64$-antenna array, which is more than $380$\% of the inter-user channel correlation in the i.i.d.\ Rayleigh fading channel.

\begin{figure*}[htbp]
  \mbox{}
  \hfill
  \begin{minipage}{.49\linewidth}
  \centering
  \subfloat[Intra-Cluster Correlation]{\label{fig:intra_cor}
                                                    \includegraphics[scale=.42]{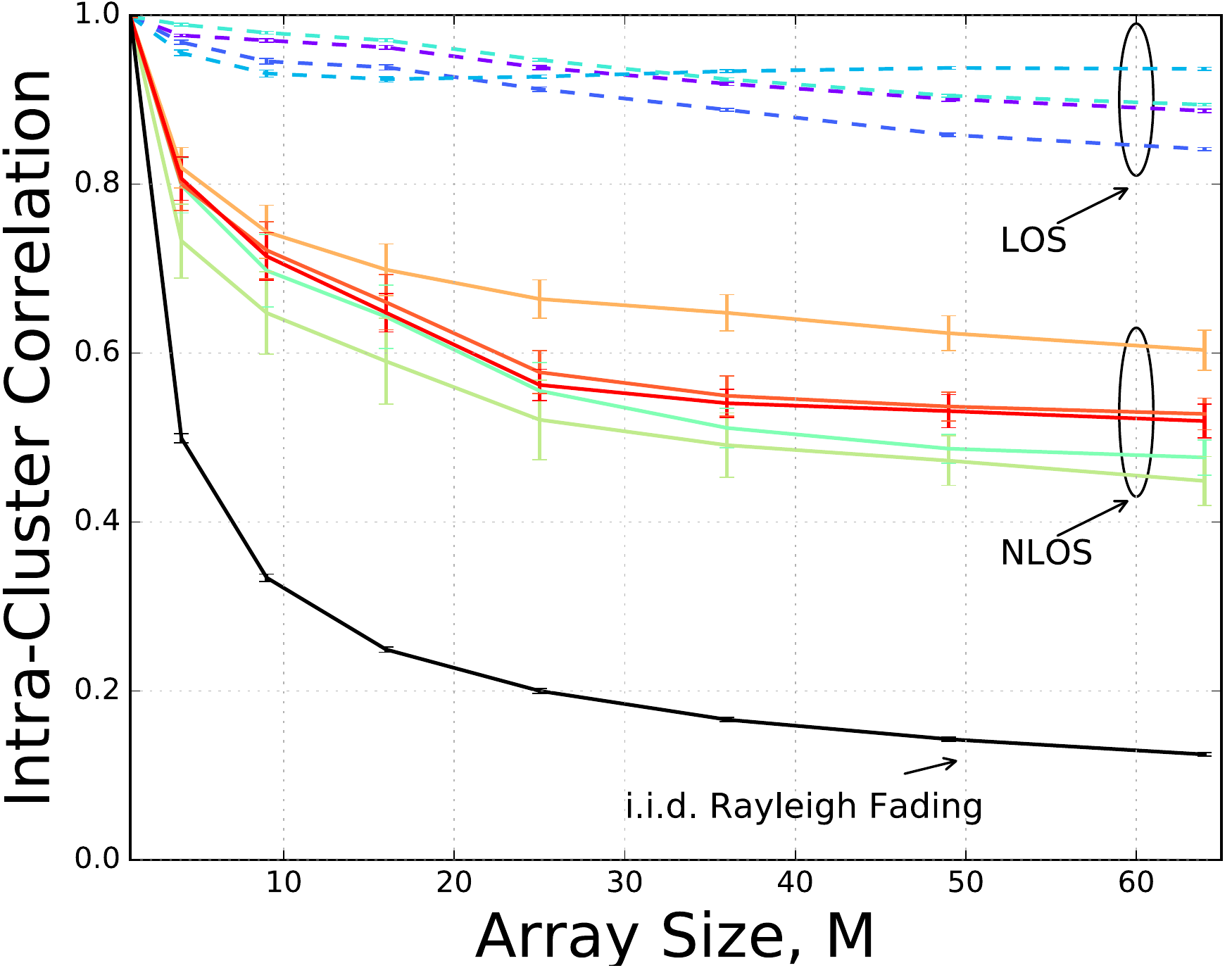}}
  \end{minipage}%
  \hfill
\begin{minipage}{.49\linewidth}
  \centering
  \subfloat[Inter-Cluster Correlation]{\label{fig:inter_cor}
                                                    \includegraphics[scale=.42]{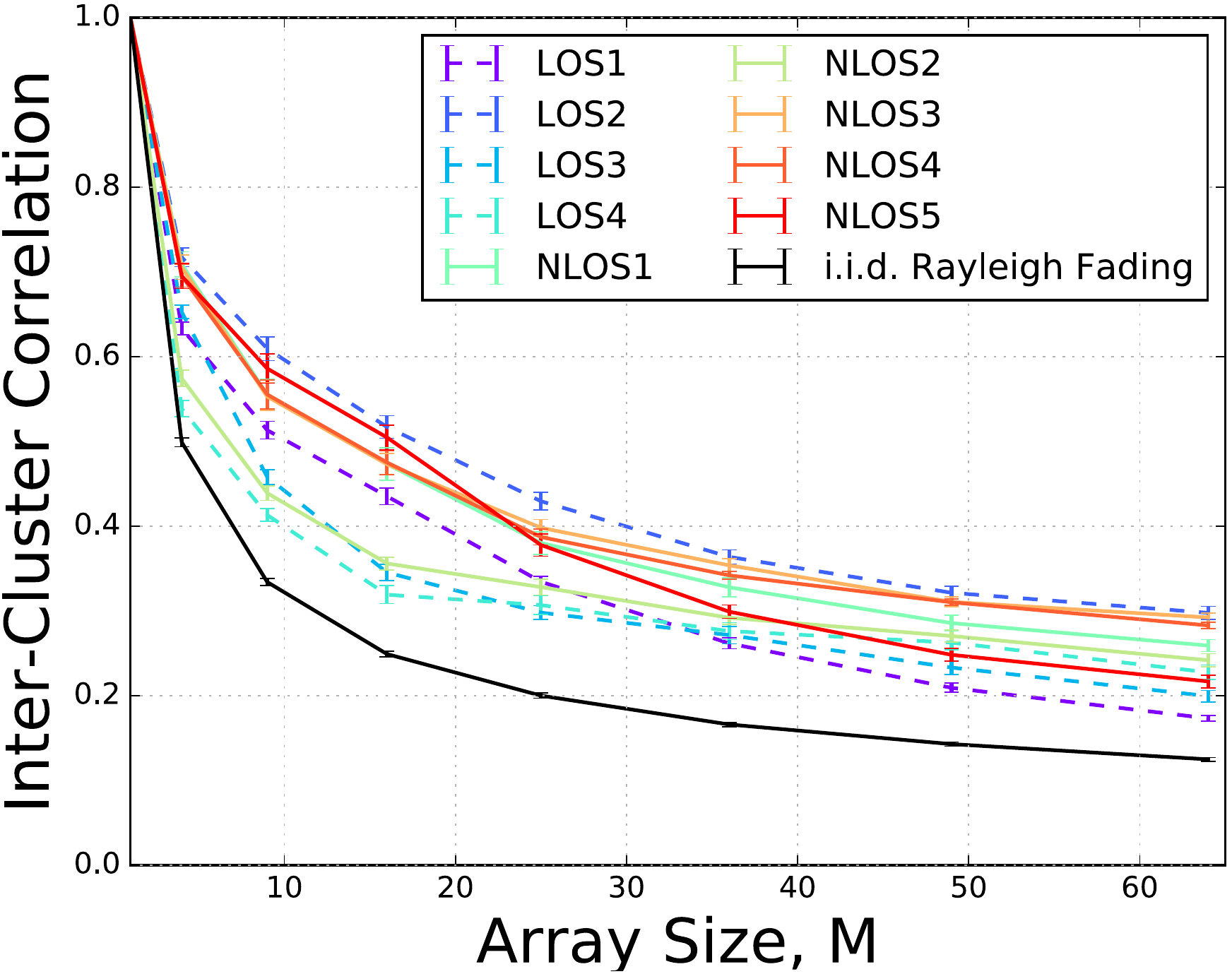}}
  \end{minipage}
  \hfill
  \mbox{}
\caption{
    Intra-cluster and inter-cluster inter-user channel correlation with the standard error bar.
    Both sub-figures share the same legend.
    The error bars denote the correlation variation across the $52$ subcarriers.
    The solid and dashed colored lines are for the NLOS clusters and the LOS clusters, respectively.
    The LOS and NLOS curves are obtained via using the inter-user channel correlation definition~\eqref{equ:alpha_def} based on measured channels.
    The measurement setup and inter-user channel correlation computation process can be found in Section III-A and Section III-B, respectively
    The solid black line is for the simulated i.i.d. Rayleigh fading channel, which is independent across subcarriers.
    For each subcarrier, the i.i.d. Rayleigh fading channel simulation used a sample size of $1000$.
}~\label{fig:correlations_VS_M}
\end{figure*}

\noindent
\emph{OTA Finding 2 -- Far-away user channels correlation decays faster than measured closed-by users, but slower than users in the i.i.d.\ Rayleigh fading channel}:
We next shift our attention to the channel correlation between users from different clusters, or {\sl inter-cluster correlation}. We first review the statistical average of the correlation between the $25$ users from each cluster to all users from the other $8$ clusters across the $52$ subcarriers.
Fig.~\ref{fig:inter_cor} provides the calculation results, and each point of every curve is acquired by averaging over $260000$ unique measured channel vector pairs. The error bars show the standard deviation across the $52$ subcarriers.
Compared to the intra-cluster correlation in Fig.~\ref{fig:intra_cor}, the inter-cluster correlation reduces at a higher rate as $M$ increases.
For both LOS and NLOS clusters, the inter-cluster channel correlation is lower than the intra-cluster correlation for all $M$ that is greater than $4$.
We remark that, despite being lower than intra-cluster correlation, the inter-cluster channel correlation is observed to be higher than the correlation in the i.i.d.\ Rayleigh fading channel.
For example, the inter-cluster correlation for NLOS $3$ with $64$ base-station antennas is $0.266$, which is over $210$\% of the inter-user channel correlation in i.i.d.\ Rayleigh fading.

\begin{figure*}[htpb]
  \centering
    \subfloat[Mean of $52$ Subcarriers]{~\label{fig:heat_mean}
  \includegraphics[width=0.5\textwidth]{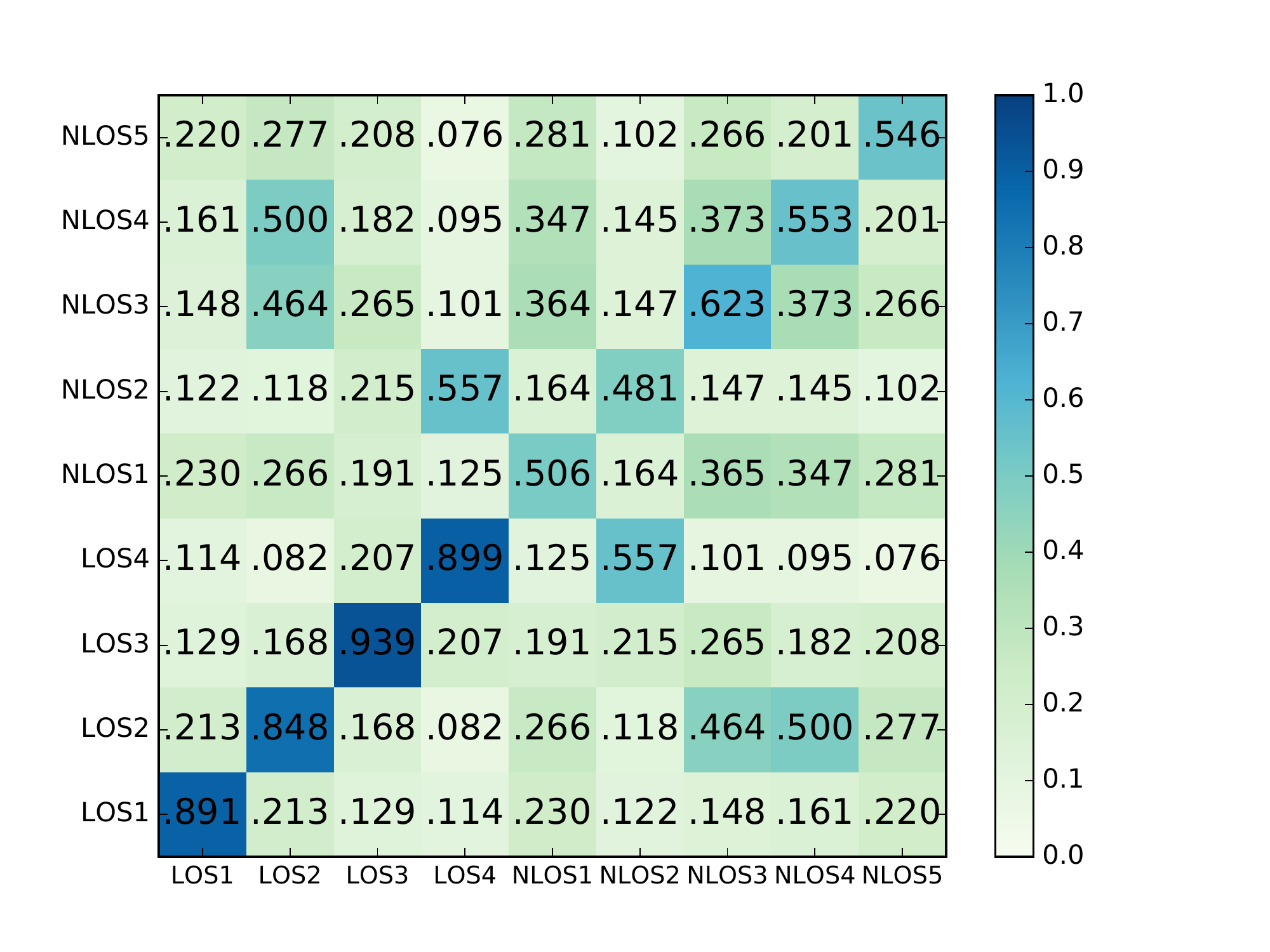}
    }
    \subfloat[Standard Deviation of $52$ Subcarriers]{~\label{fig:heat_std}
  \includegraphics[width=0.5\textwidth]{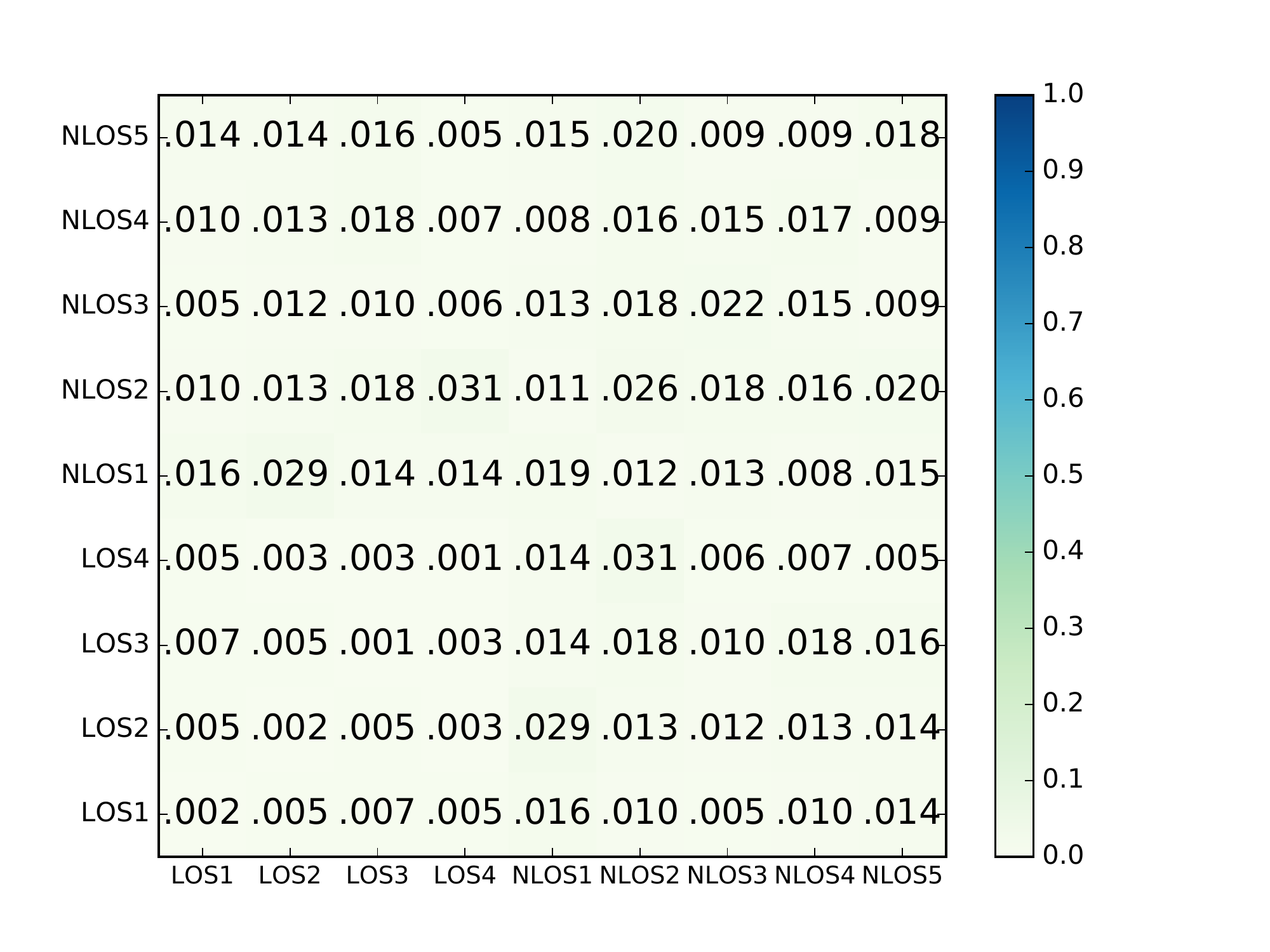}
  }
  \caption{Inter-user channel correlation mean and standard deviation of measured channels with $M=64$ between all pairs of clusters.
  Each grid is the inter-user channel correlation average or standard deviation between a cluster labeled on the $x$-axis and a cluster labeled on the $y$-axis.}~\label{fig:cor-heat-map}
\end{figure*}

We now examine the mean and standard deviation of the channel correlation between users from any pair of clusters.
In Figure~\ref{fig:heat_mean}, each grid of the matrix denotes the inter-user channel correlation between a different cluster pair when $M=64$.
The diagonal elements denote the intra-cluster channel correlation (each averaged over $15600$ channel vector pairs).
And the off-diagonal elements denote the inter-cluster correlation (each averaged over $32500$ channel vector pairs).
Interestingly, we find that the inter-cluster correlation of some cluster pairs is still significant even with a $64$-antenna base-station. For example, between Cluster LOS $4$ and Cluster NLOS $2$, the inter-cluster correlation is $0.557$.
And $30.56$\% of the $36$ cluster pairs have inter-cluster correlation that is at least twice higher than the correlation in the i.i.d. Rayleigh fading.
In summary, the statistical average of all far-away users' channel correlation reduces as the number of base-station antennas increases.
However, some far-away user pairs can still be of high correlation with a $64$-antenna base-station.

\noindent
\emph{OTA Finding 3 -- Inter-user angle proximity might cause high inter-user channel correlation}:
To provide insights into the measured high channel correlation between LOS $4$ and NLOS $2$, we now present spatial angle estimation examples of four measured users.
The four examples consist of two users from NLOS $2$, one from NLOS $3$, and one from LOS $4$.

We first focus on Fig.~\ref{fig:aoa-nlos2-1} and~\ref{fig:aoa-nlos2-2} that present the estimated angle energy map of two users from NLOS $2$. The two users demonstrate proximity in the angle space near $\left(28^{\circ}, 0^{\circ}\right)$ and $\left(15^{\circ}, 15^{\circ}\right)$, which might explain the observed $0.481$ intra-cluster correlation when $M=64$ in NLOS $2$.
Similarly, the channel correlation of $0.557$ between NLOS $3$ and LOS $2$ could be explained by the small angle difference in Fig.~\ref{fig:aoa-nlos2-1},~\ref{fig:aoa-nlos2-2},~\ref{fig:aoa-los4}.
And the large angle separations in Fig.~\ref{fig:aoa-nlos3} and Fig.~\ref{fig:aoa-nlos2-2} could explain the observed lower ($0.147$) inter-cluster correlation between NLOS $2$ and NLOS $3$.
For the interested readers, the measured channel dataset and our implementation of MUSIC~\cite{MUSIC-review} can be downloaded from~\cite{open_data_link}, which will allow exploring channels from $>200$ locations.
To understand the impact of the inter-user angle proximity, Section~\ref{sec:RayAnalysis} will compute the inter-user channel correlation as a function of the inter-user angle correlation in closed-form.

\begin{figure*}[htbp]
\centering
  \subfloat[NLOS 2, Location 1]{ \label{fig:aoa-nlos2-1}
\includegraphics[height=0.21\textwidth]{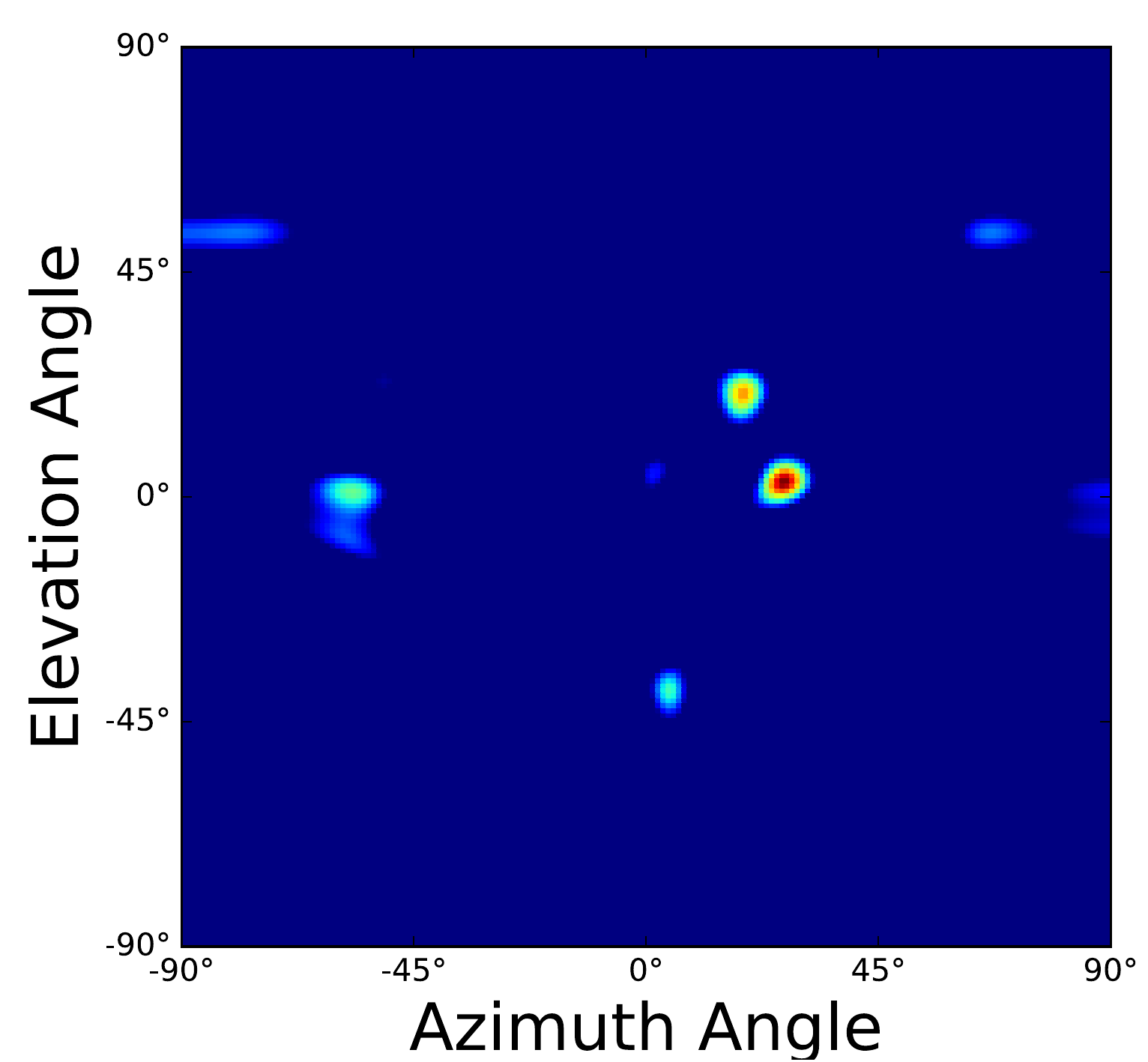}
  }
  \subfloat[NLOS 2, Location 2]{ \label{fig:aoa-nlos2-2}
\includegraphics[height=0.21\textwidth]{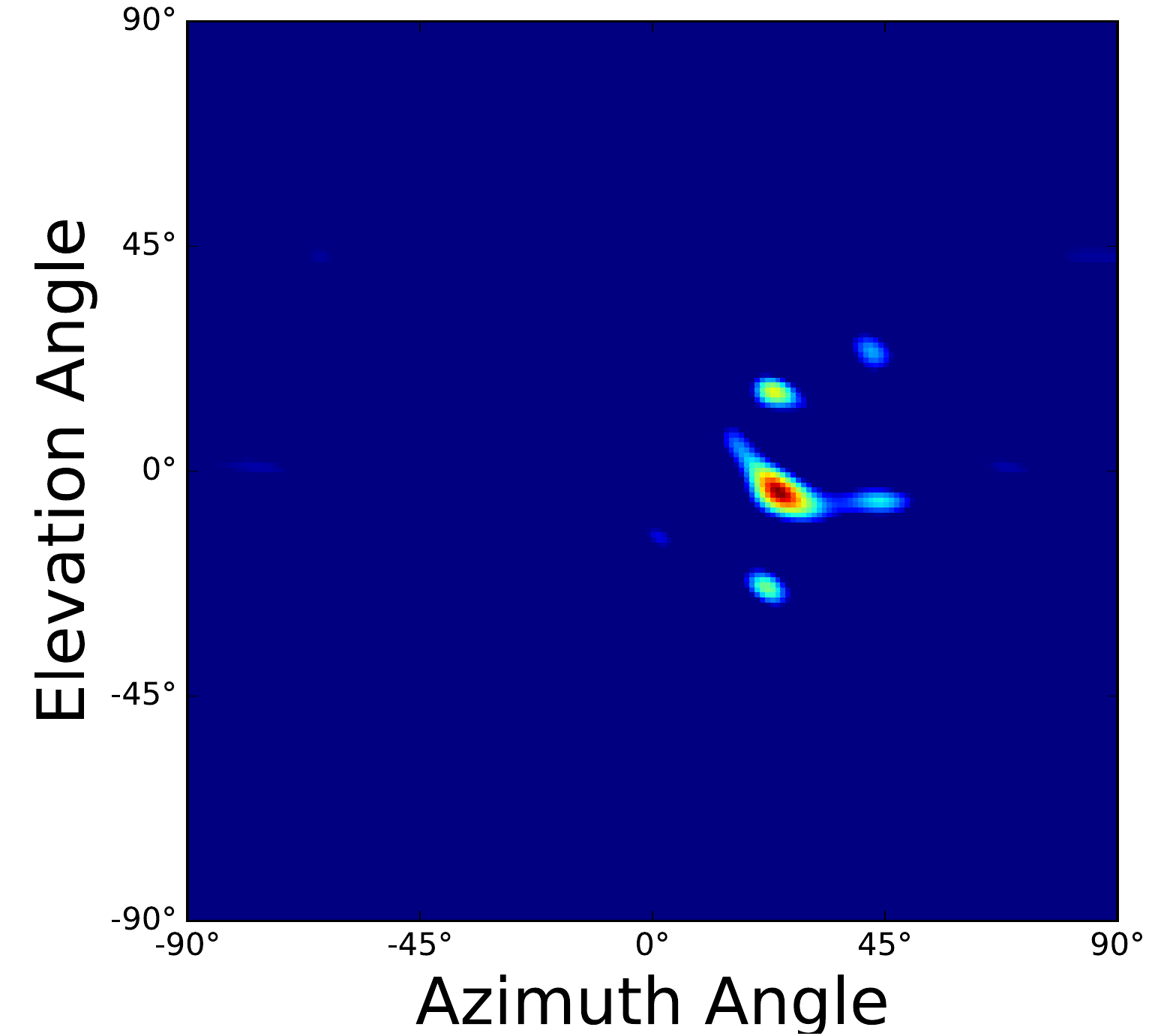}
  }
  \subfloat[NLOS 3, Location 1]{ \label{fig:aoa-nlos3}
\includegraphics[height=0.21\textwidth]{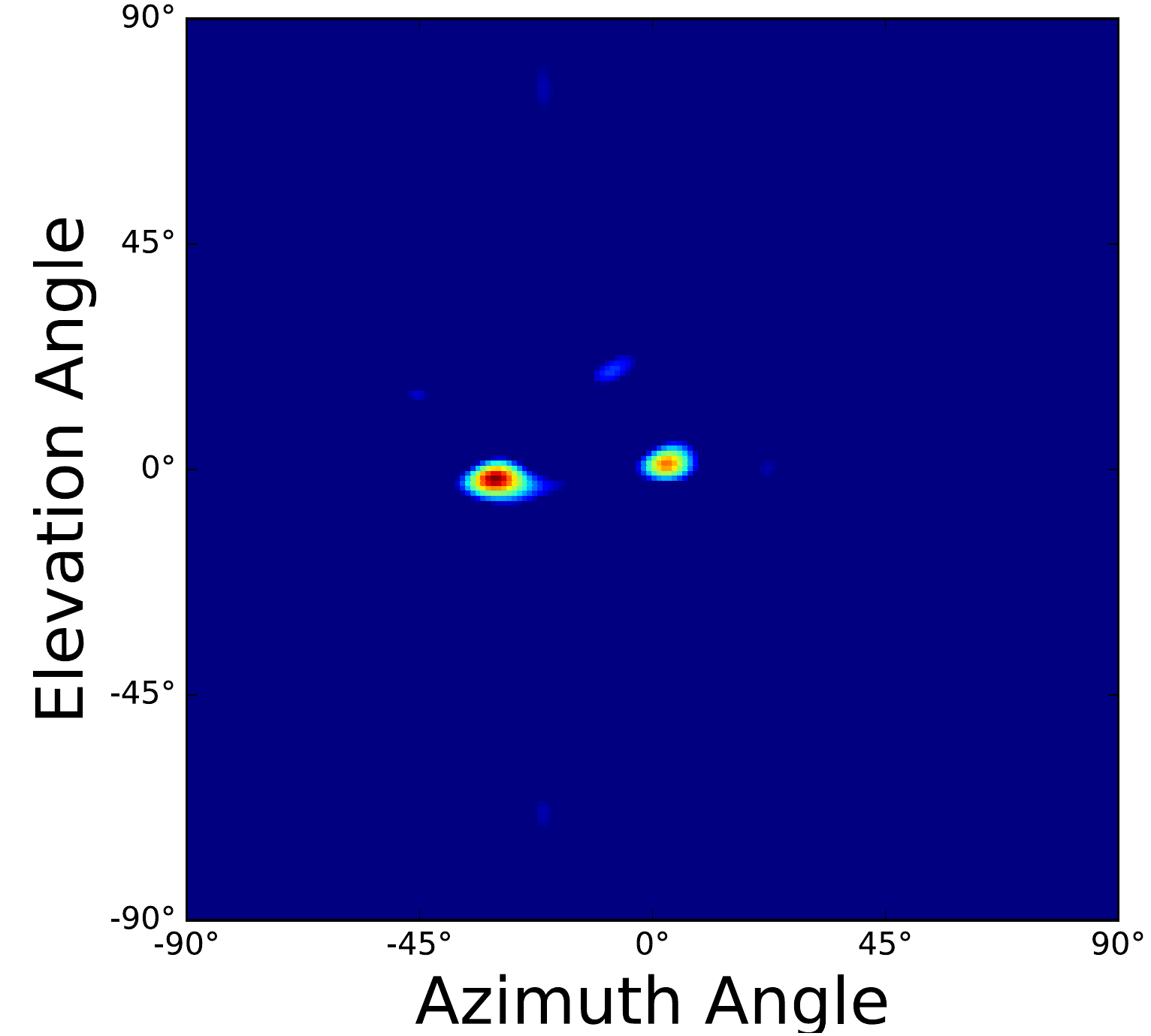}
  }
  \subfloat[LOS 4, Location 1]{ \label{fig:aoa-los4}
\includegraphics[height=0.21\textwidth]{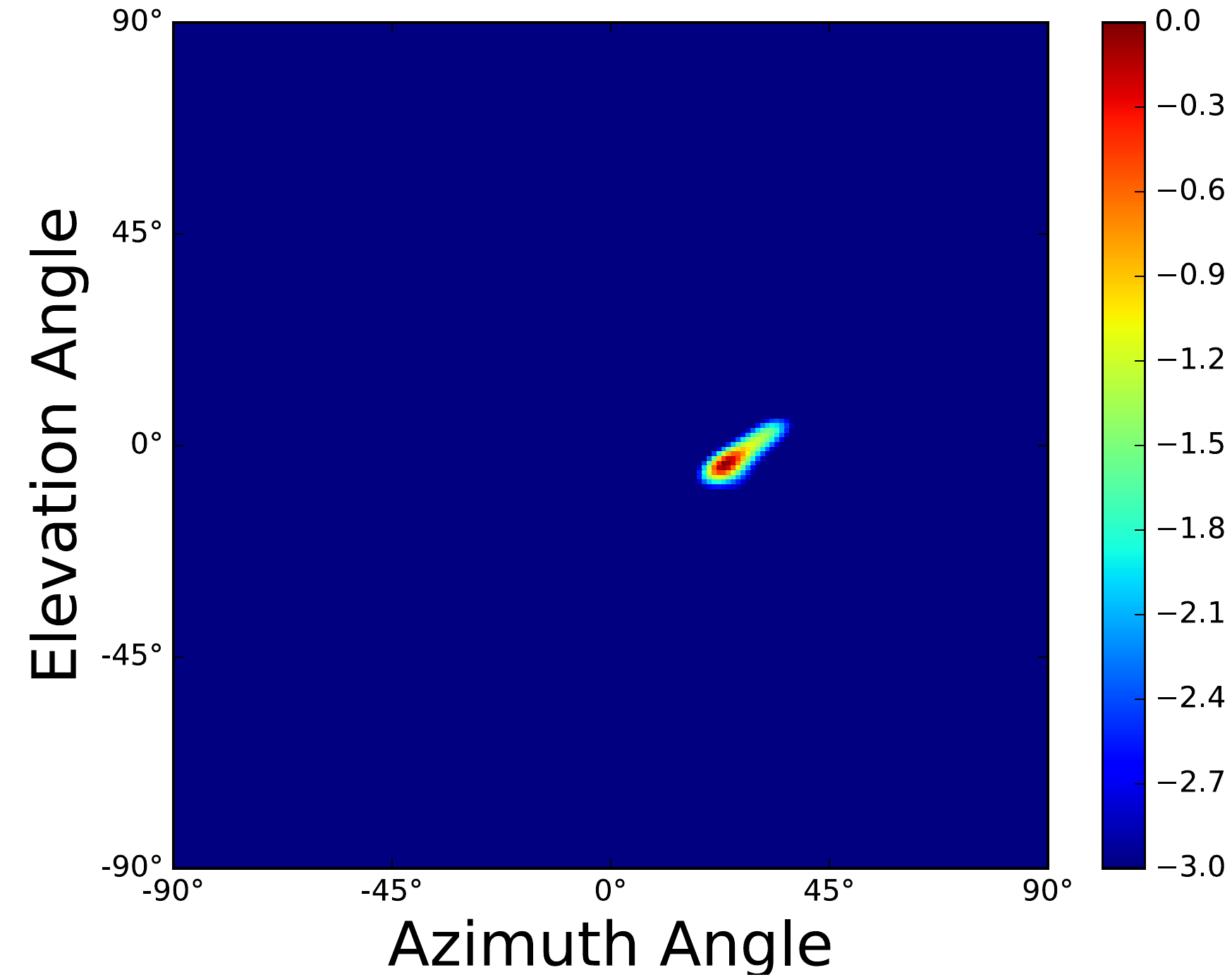}
  }
\caption{The estimated power angle map of different users via using MUSIC~\cite{MUSIC-review} based on the measured channels. Each point of each figure approximates the logarithmic energy of each direction.
The source code and the measured channel dataset are available in~\cite{open_data_link}.
}~\label{fig:aoa}
\end{figure*}

\noindent
\emph{OTA Finding 4 -- Inter-user channel correlation is near-constant across the subcarriers}:
We profiled the channel correlation variance over different frequency subcarriers.
The error bars in Fig.~\ref{fig:correlations_VS_M} denote the variation of the intra-cluster and inter-cluster correlations across all $52$ subcarriers.
Each grid of Fig.~\ref{fig:heat_std} presents the standard deviation of the average channel correlation when $M=64$.
The small error bars in Fig.~\ref{fig:correlations_VS_M} and small standard deviation values in Fig.~\ref{fig:heat_std} hence demonstrate that the intra-cluster and inter-cluster correlation is near-constant across the $20$ MHz band.

\noindent
\emph{OTA Finding 5 -- Close-by user channel correlation reduces with multi-path.}
By comparing intra-cluster user channel correlation in Fig.~\ref{fig:intra_cor}, we find that the intra-cluster channel correlation in LOS clusters is higher than that in NLOS clusters.
Fig.~\ref{fig:intra_cor} shows that the lowest observed LOS intra-cluster channel correlation ($0.838$) is higher than the highest NLOS intra-cluster channel correlation ($0.623$).

In our measured dataset, close-by users' high channel correlation barely reduces when $M>36$.
Furthermore, $30.56$\% of far-away user pairs have channel correlation that is over $200$\% higher than the correlation in the i.i.d. Rayleigh fading channel.
Besides, we observe that users with high channel correlation can have proximity in the angle space.

\section{User Correlation: A Spatial Perspective}\label{sec:RayAnalysis}
Section~\ref{sec:measurement} presents channel measurements that capture the inter-user channel correlation in a real-world propagation environment.
The spatial angle estimation examples in Fig.~\ref{fig:aoa} demonstrate that angles of different users can be close-by.
For users with independently distributed angles and various base-station array configurations, past theoretical analysis~\cite{favorable_prop, wu2017favorable} proved that inter-user channel correlation quickly reduces with the number of base-station antennas.
This section analyzes how inter-user angle correlation and base-station array configuration impact the inter-user channel correlation.

This section begins with the adopted spatial channel model.
Section~\ref{subsec:main_the} presents the main propositions.
The rest of the section utilizes specialized versions of the main propositions to explain the measured inter-user channel correlation.

\subsection{Spatial Channel Model}~\label{subsec:channel_model}
To model the massive MIMO channels with inter-user angle correlation, we consider a multiple-path spatial channel model, which is illustrated in Fig.~\ref{fig:cha_mdl}.
The base-station is equipped with an $M$-antenna planar array on the $x-y$ plane, whose setup is presented in Fig.~\ref{fig:planar}.
\begin{figure}[htbp]
  \centering
  \includegraphics[width=0.33\textwidth]{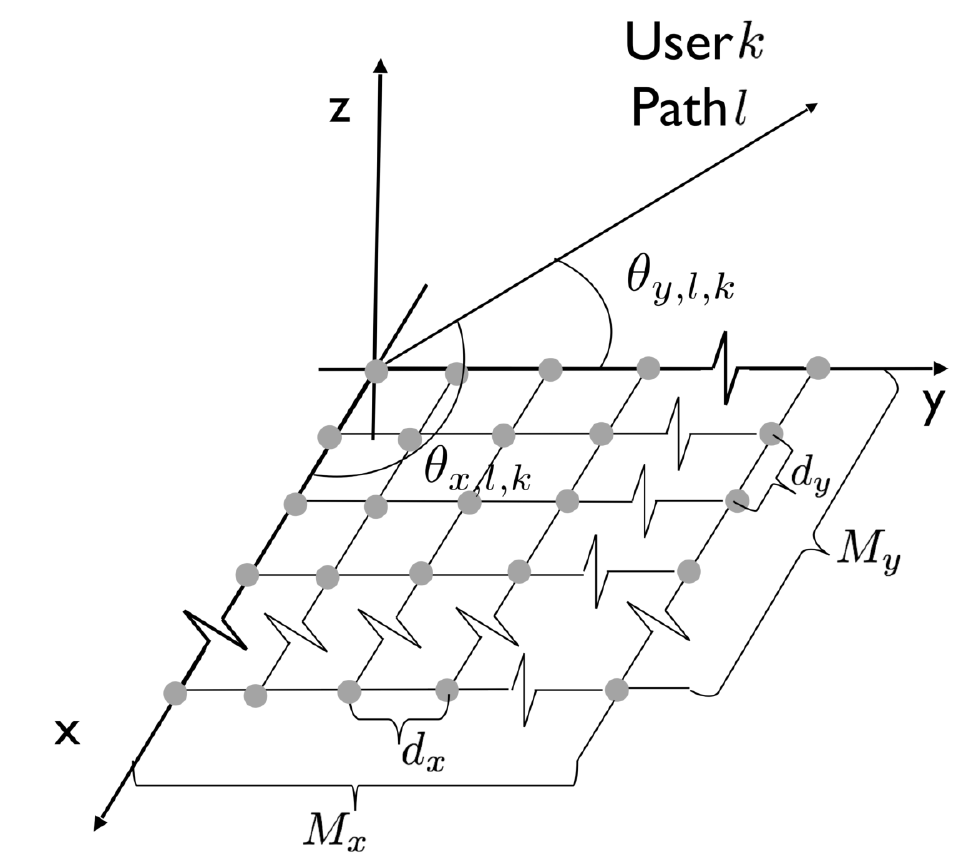}
  \caption{The setup of the considered 2-D planar array. The array is placed on the $x-y$ plane. There are $M_{x}$ ($M_{y}$) antennas on the $x$-axis ($y$-axis). The antenna spacing along the $x$-axis ($y$-axis) is $d_{x}$ ($d_{y}$).}~\label{fig:planar}
\end{figure}
On the $x$-axis, there are $M_{x}$ antennas uniformly spaced with spacing of $d_{x}$.
In each row, there are $M_{y}$ antennas uniformly spaced on $y$-axis with spacing of $d_{y}$.
The base-station is elevated high above the ground for a typical cellular network and free from local scatterers.
The base station, therefore, is in the far-field.
And the channel of User $k$ equals
\begin{equation}
\mathbf{h}_{k} = \sum_{l=1}^{L} \beta_{k,l} e^{j\phi_{k,l}} \mathbf{a}\left(\bm{\theta}_{k,l}\right),~\label{equ:ray_mdl}
\end{equation}
where $L$ is the number of paths,
$l$ is the index of the spatial path,
$\beta_{k,l}$ denotes the path gain,
$\phi_{k,l}$ is path phase,
and $\mathbf{a}\left(\bm{\theta}_{k,l}\right)$ is the array response vector for a plane wave at direction $\bm{\theta}_{k,l}$.
The array response vector is jointly determined by the carrier frequency and the array configuration.
As we are interested in explaining the measured channels in Section~\ref{sec:measurement}, we consider the case of the Uniform Planar Array (UPA), whose response vector is~\cite[Chapter 7]{bjornson2017massive}
\begin{align}
\mathbf{a}\left(\bm{\theta}\right) =
\Big[& 1, \
\dots,
e^{j\frac{2\pi}{\lambda}
\left( d_{x} m_{x} \sin \theta_{x} +
    d_{y} m_{y} \sin \theta_{y} \right)}, \
\dots, \notag \\
& e^{j\frac{2\pi}{\lambda}
\left( d_{x} \left(M_{x}-1\right) \sin \theta_{x} +
    d_{y} \left(M_{y}-1\right) \sin \theta_{y} \right)
  }\Big]
,~\label{equ:def_a_theta_UPA}
\end{align}
where $d_{x}$ is the the antenna spacing in the $x$ axis,
$d_{y}$ is the antenna spacing in the $y$ axis,
$m_{x}$ is the antenna index in the $x$ axis,
$m_{y}$ is the antenna index in the $y$ axis,
$\lambda$ is the wavelength of the subcarrier,
$\bm{\theta}=\left[\theta_{x}, \theta_{y}\right]$ denotes the spatial angles,
$\theta_{x}$ is the angle between $x$ axis and the AoD,
and $\theta_{y}$ is the angle between $y$ axis and the AoD.

Under channel model~\eqref{equ:ray_mdl}, the inter-user channel correlation~\eqref{equ:alpha_def} is determined by the joint distribution of user path phases, path gains, and path angles.
As the user to base-station distance is usually far larger than the wavelength, the path phases $\phi_{k,l}$ follow i.i.d. uniform distribution over $\left[0,2\pi\right]$.
In this section, we consider constant path gains, and the path angles are correlated between users.

   The multi-user angle joint distribution depends on the propagation environment, carrier frequency, user location, and user mobility.
To the best of our knowledge, the analytical modeling of multi-user angle distribution in practical massive MIMO downlink systems remains open.
This section models inter-user angle correlation of User $k$ Path $l$ and User $k^{'}$ Path $l^{'}$ in the $x$-axis by letting $\sin\big(\theta_{x,k,l}\big)- \sin\big(\theta_{x,k^{'},l^{'}}\big)$ follow uniform distribution over $\left[-\alpha_{x,l, l^{'}}, \alpha_{x,l, l^{'}}\right]$.
  Similarly, $\sin\big(\theta_{y,k,l}\big)- \sin\big(\theta_{y,k^{'},l^{'}}\big)$ follows independent uniform distribution over $\left[-\alpha_{y,l, l^{'}}, \alpha_{y,l, l^{'}}\right]$.
  In this model, constants $\alpha \in \left[0,1\right]$ correspond to the case of maximum inter-user angle separation.
  For any distribution with support on $\left[-\alpha, \alpha\right]$, it is known~\cite{cover1999elements} that the uniform distribution is entropy maximizing and thus minimizes the prior information of the joint multi-user angle distribution.
    A smaller $\alpha$ captures a stronger angle correlation.
  And  $\alpha_{x,l, l^{'}}=0$ suggests the corner case where the $\theta_{x,k,l}$ equals $\theta_{x,k^{'},l^{'}}$ with probability $1$.

 The uniform distribution model is also adopted in existing inter-user channel correlation analysis without inter-user angle correlation~\cite{marzetta2016fundamentals, bjornson2017massive, favorable_prop, wu2017favorable, masouros2015space, yang2017massive}.
 And the existing analysis can be viewed as specialized versions of this section's analysis by letting $\alpha_{x,l,l^{'}}=1$, $M_y=1$, $L=1$,  and $N=1$.
 We will later compare this section's analysis with the existing results to highlight the impact of inter-user angle correlation.
We consider it a future direction to analyze the inter-user channel correlation based on modeling the inter-user angle correlation in real-world environments.

\begin{myremark}
The covariance matrix method~\cite{adhikary2013joint},~\cite[(2.19)]{bjornson2017massive} can also model the joint user channel distribution by assuming near-static second-order covariance matrices.
The elements of the covariance matrices capture the correlations among the base-station antennas. The covariance matrix method is prevalent in the massive MIMO achievable rate characterizations.
This paper chooses the spatial channel model for two reasons.
The spatial channel model~\cite{marzetta2016fundamentals, bjornson2017massive, 3gpp.36.873} can capture the effect of both base-station array configuration and inter-user angle correlation.
Propositions in Section~\ref{subsec:main_the} will presents the inter-user channel correlation as a function of inter-user angle correlation and the base-station array configuration in closed-form.
Besides, most of the past massive MIMO inter-user channel correlation analysis adopts spatial channel models (without inter-user angle correlation).
Therefore, the spatial channel model enables a comparison to highlight the impact of inter-user angle correlation.
\end{myremark}

\subsection{Main Proposition}~\label{subsec:main_the}
We now present the main propositions that capture the inter-user channel correlation with inter-user angle correlation.
This subsection begins with the finite-array analysis, which is followed up by the large-array regime specification.
\begin{Proposition}[Inter-user Angle Correlation Increases Inter-user Channel Correlation]~\label{prop:main_finite}
For massive MIMO systems with a uniform planar array and in channel model~\eqref{equ:ray_mdl}, the square of the inter-user correlation~\eqref{equ:alpha_def} satisfies
\begin{align}
\Var \left[\frac{1}{M} \mathbf{h}_{k}^{H}\mathbf{h}_{k^{'}}\right]
=  \sum_{l=1}^{L} \sum_{l^{'}=1}^{L}  &  \left(\beta_{k,l} \beta_{k^{'},l^{'}}\right)^2\Big[\eta\left(d_{x} \alpha_{x,l,l^{'}}/\lambda , M_{x}\right) \notag \\
& \eta \left(d_{y} \alpha_{y,l,l^{'}}/\lambda, M_{y}\right)\Big], \label{equ:main_finite_var}
\end{align}
where
$\eta \left(c, M\right)= \frac{1}{M^2}\sum_{m_{1}=0}^{M-1}\sum_{m_{2}=0}^{M-1} \sinc\left( 2\pi c \left(m_{1}-m_{2} \right)\right)$.
\end{Proposition}

\begin{proof}
  The main proof is to use the channel definition~\eqref{equ:ray_mdl} and decompose the target value into the contributions by inter-user path pairs.
  The rest of the proof is completed by computing the mean and the variance of channel dot product of two spatial paths with correlated angles. Appendix~\ref{appendix:main_finite} provides the complete proof.
\end{proof}

Proposition~\ref{prop:main_finite} captures the inter-user channel correlation with inter-user angle correlation in closed-form.
Proposition~\ref{prop:main_finite} shows that the inter-user channel correlation is determined by the inter-user angle correlation and the base-station array configuration, which includes the number of antennas and inter-antenna spacing.
Due to the $\sinc$ sum, the impact of the inter-user angle correlation and the base-station antenna spacing on inter-user channel correlation is not immediate, which will be clarified by the following asymptotic specification.

\begin{Proposition}[Asymptotic Inter-user Channel Correlation]~\label{prop:main_asy}
  For a massive MIMO system with a square uniform planar array, where $M_{x}=M_{y}=\sqrt{M}$, the square of the inter-user channel correlation~\eqref{equ:alpha_def} satisfies the following as $M\to\infty$
  \begin{align}
 \Var \left[\frac{1}{M} \mathbf{h}_{k}^{H}\mathbf{h}_{k^{'}}\right] \cong
 \sum_{l=1}^{L} \sum_{l^{'}=1}^{L}  &  \left(\beta_{k,l} \beta_{k^{'},l^{'}}\right)^2\Big[\tilde{\eta}\left(d_{x} \alpha_{x,l,l^{'}}/\lambda , M_{x}\right) \notag \\
& \tilde{\eta} \left(d_{y} \alpha_{y,l,l^{'}}/\lambda, M_{y}\right)\Big],
\end{align}
where $\tilde{\eta}\left(c, M\right)=\bm{1}_{c = 0} + \bm{1}_{c \neq 0}  \frac{1}{2 c M}$. Here, $f(M)\cong g(M)$ denotes the asymptotic equivalence relation~\cite[1.4]{de1981asymptotic}, which holds if and only if $\lim_{M\to\infty} f(M)/g(M)= 1$.
\end{Proposition}

\begin{proof}
Appendix~\ref{appendix:asy_proof} presents the proof, which is by combining Proposition~\ref{prop:main_finite} and Lemma~\ref{append:lemma_2d_sinc_sum} in Appendix~\ref{append:user_cor}.
\end{proof}

Proposition~\ref{prop:main_asy} captures the scaling of the inter-user channel correlation in the large-array regime. By definition, $\cong$ denotes asymptotic equivalence.
Proposition~\ref{prop:main_asy} considers the square array for symbolic simplification. Similar results can be extended for UPAs where $M_{x}$ and $M_{y}$ grows at a fixed rate as $M$ goes to infinite.

The square of the inter-user channel correlation is the linear sum of the contributions from each inter-user path pair. For path pair that consists of User $k$ Path $l$ and User $k^{'}$ Path $l^{'}$,
the contribution is $\left(\beta_{k,l} \beta_{k^{'},l^{'}}\right)^2
 \tilde{\eta} \left(d_{x}\alpha_{x,l,l^{'}}/\lambda, M_{x} \right)
 \tilde{\eta} \left(d_{y}\alpha_{y,l,l^{'}}/\lambda , M_{y}\right)$.

When $\alpha_{x,l,l^{'}}$ and  $\alpha_{y,l,l^{'}}$ equal $0$, the AoD of User $k$ Path $l$ equals the AoD of User $k^{'}$ Path $l^{'}$ with probability $1$, which means a always shared angle.
Such shared angle leads to inter-user channel correlation converges to a positive constant that is independent of $M$.
Hence, the inter-user channel correlation does not converge to zero in the large-array regime when a shared angle exists.

For a more practical case where $\alpha_{x,l,l^{'}}>0$ or $\alpha_{y,l,l^{'}}>0$  for all path pairs, the inter-user channel correlation does reduce with the number of antennas at a rate of $\frac{\lambda}{ d_{x}\alpha_{x,l,l^{'}} M_{x}}
\frac{\lambda}{d_{y} \alpha_{y,l,l^{'}}M_{y}} $.
The inter-user channel correlation increases with wavelength and inter-user angle correlation.
And a larger inter-antenna spacing or more base-station antennas leads to a smaller inter-user channel correlation.

Propositions~\ref{prop:main_finite} and~\ref{prop:main_asy} further capture the interplay between inter-user angle correlation and base-station antenna spacing. For the $x$ ($y$) axis, the channel correlation contribution term is determined by the product $\alpha_{x} d_{x}$ ($\alpha_{y} d_{y}$), where the angle correlation factor serves as a scaling constant before the base-station antenna spacing.
This observation suggests that one can adapt the antenna spacing $d$ inversely proportional to $\alpha$ to achieve the exact inter-user channel correlation.
In other words, manipulating the base-station array configuration is an attractive option to reduce inter-user channel correlation.

Another practical setup is the space-contained base-station, where the base-station inter-antenna spacing reduces inversely proportionally to $M$ due to limited base-station space.
For the space-constrained systems, the $d_{x}M_{x}$ ($d_{y} M_{y}$) is a fixed constant.
And Proposition~\ref{prop:main_asy} proves that the inter-user channel correlation would converge to a positive constant, which is different from the fixed inter-antenna spacing case.
With fixed inter-antenna spacing, the inter-user channel correlation will converge to zero as the number of base-station antenna $M$ increases to infinity.
Therefore, the limited base-station array aperture in space-constrained systems can lead to high inter-user channel correlation even for very large $M$.

Propositions~\ref{prop:main_finite} and~\ref{prop:main_asy} characterize the impact of the inter-user angle correlation and base-station array configuration on inter-user channel correlation in the finite-array and the large-antenna regimes. Section~\ref{sec:Numerical} uses numerical experiments to validate the finite-array analysis and examine the large-array analysis accuracy.

\subsection{Case Studies}~\label{subsec:theory_case_study}
This subsection leverages Propositions~\ref{prop:main_finite} and~\ref{prop:main_asy} to explain the measured OTA channels and past theoretical results.
We first employ the propositions to explain the OTA channel observations in Section~\ref{sec:measurement}.
The rest of this subsection connects our characterization to past results on inter-user channel correlation without inter-user angle correlation.

\subsubsection{OTA Observations Explanations}

\noindent
\emph{OTA Finding 1 -- Close-by user channel correlation decays very slowly.}
Past research demonstrated~\cite{kolmonen2010measurement,liu2012cost} that close-by users are of a higher probability to have close-by angles due to shared scatters.
Therefore, it is reasonable to consider that the inter-user angle correlation reduces with the inter-user distance.
Propositions~\ref{prop:main_finite} and~\ref{prop:main_asy} demonstrate that the inter-user channel correlation increases with the inter-user angle correlation in both the finite-array and large-array regimes.
Therefore, the combination of the theoretical results and past observations explains the first OTA finding.

\noindent
\emph{OTA Finding 2 -- Far-away user channels correlation decays faster than measured closed-by users, but slower than users in the i.i.d.\ Rayleigh fading channel.}
In practical environments, two clusters of far-away users may share the main energy from similar angles.
In addition, the measurement base-station employed patch antennas with $3$-dB beam-width of around $120$-degree.
Therefore, the angle correlation of users from far-away clusters could be significant, which leads to a high channel correlation among some far-away users.

\noindent
\emph{OTA Finding 3 -- Inter-user angle proximity might cause high inter-user channel correlation.}
Under the considered channel model~\eqref{equ:ray_mdl}, inter-user angle proximity means high inter-user angle correlation. And Propositions~\ref{prop:main_finite} and~\ref{prop:main_asy} prove that inter-user channel correlation increases with inter-user angle correlation.

\noindent
\emph{OTA Finding 4 -- Inter-user correlation is near-constant across the subcarriers.}
Proposition~\ref{prop:main_finite} describes the inter-user channel correlation.
In~\eqref{equ:main_finite_var}, the impact of the frequency is captured by the carrier wavelength.
The central carrier frequency is $2.4$ GHz on the measured ISM band with a frequency bandwidth of $20$ MHz.
Therefore, the lowest and highest subcarriers will differ by $0.83$\% in wavelength, which is part of the $\sinc$ function input.
The smaller $\sinc$ function input difference explains the near-constant inter-user channel correlation across the measured subcarriers.
To better understand the implication, we remove the $\sinc$ function by considering the large-array specification.
Proposition~\ref{prop:main_asy} shows that the lowest subcarrier will have $1.68$\% higher channel correlation than the highest subcarrier in the large-array regimes, which matches with the OTA Finding $4$.

We note that the above cross-subcarrier analysis assumes the same spatial paths across different subcarriers.
This assumption is widely adopted in the state-of-the-art channel modeling~\cite{3gpp.36.873} for TDD-based systems.
For FDD-based massive MIMO systems~\cite{marzetta2016fundamentals,bjornson2017massive,zhang2018directional}, the uplink and the downlink channel central frequencies can have a large frequency gap.
The gap indicates that the spatial path difference might lead to a more considerable difference in inter-user channel correlation between the uplink and the downlink channels.

\noindent
\emph{OTA Finding 5 -- Close-by user channel correlation reduces with multi-path.}
To explain the impact of the multi-path effect, we consider the keyhole channels~\cite{marzetta2016fundamentals} with a different number of keyholes.
The keyhole channel can be a special case of model~\eqref{equ:ray_mdl} by letting each user channel has a path passing each keyhole.
For each keyhole, inter-user angle correlation is captured by $\bm{\alpha}_{o}=\left(\alpha_{x,o}, \alpha_{y,o}\right)$.
The inter-user angle correlation of paths through different holes adopts $\alpha=1$ to model the independence of keyholes.
To simplify the expression, we use $\beta_{k,l}=\beta_{k^{'},l^{'}}=\frac{1}{\sqrt{L}}$.

Proposition~\ref{prop:main_finite} can give the inter-user channel correlation with the keyhole channel specifications. After some algebraic manipulations, the inter-user channel correlation~\eqref{equ:alpha_def} satisfies
\begin{align}
  \Var \left[\frac{1}{M} \mathbf{h}_{k}^{H}\mathbf{h}_{k^{'}}\right]
= & \Big[\frac{1}{L} \eta \left(\frac{d_{x} \alpha_{x,o}}{\lambda}, M_{x}\right) \eta \left(\frac{d_{y} \alpha_{y,o}}{\lambda}, M_{y}\right) + \notag \\
 & \frac{L-1}{L} \eta\left(\frac{d_{x} }{\lambda} , M_{x}\right)  \eta \left(\frac{d_{y}}{\lambda}, M_{y}\right)\Big], \label{equ:keyhole_finite}
\end{align}
where $\eta\left(c, M\right)$ is defined in Proposition~\ref{prop:main_finite} as $\sum_{m_{1}=0}^{M-1}\sum_{m_{2}=0}^{M-1} \sinc\left( 2\pi c \left(m_{1}-m_{2} \right)\right)$.
For the keyhole channel, the inter-user channel correlation can be viewed as a weighted sum of two terms.
The first term $\eta \left(\frac{d_{x} \alpha_{x,o}}{\lambda}, M_{x}\right) \eta \left(\frac{d_{y} \alpha_{y,o}}{\lambda}, M_{y}\right)$ is the inter-user channel correlation of a LOS channel with angle correlation of $\bm{\alpha}_{o}$. And $\eta\left(\frac{d_{x} }{\lambda} , M_{x}\right)  \eta \left(\frac{d_{y}}{\lambda}, M_{y}\right)$
is the inter-user channel correlation of a LOS channel with independent angles.
Therefore, as $L$ increases, the inter-user channel correlation will be weighted towards the independent angle case, which results in a lower inter-user channel correlation.

We next extend the observation to the asymptotic regime with $M_{x}=M_{y} =\sqrt{M}$. By Proposition~\ref{prop:main_asy}, the inter-user channel correlation satisfies
\begin{align}
  \Var \left[\frac{1}{M} \mathbf{h}_{k}^{H}\mathbf{h}_{k^{'}}\right]
\cong & \Big[\frac{1}{L} \tilde{\eta} \left(\frac{d_{x} \alpha_{x,o}}{\lambda}, M_{x}\right) \tilde{\eta} \left(\frac{d_{y} \alpha_{y,o}}{\lambda},  M_{y}\right) + \notag \\
& \frac{L-1}{L} \frac{\lambda}{2 d_{x}  M_{x}} \frac{\lambda}{2 d_{y}  M_{y}}\Big], \label{equ:keyhole_asy}
\end{align}
where $\tilde{\eta}\left(c, M\right)=\bm{1}_{c = 0} + \bm{1}_{c \neq 0}  \frac{1}{2 c M}$, and $f(M)\cong g(M)$ means asymptotic equivalence.
For the above keyhole channel, both~\eqref{equ:keyhole_finite} and~\eqref{equ:keyhole_asy} show that the multi-path effect makes the inter-user channel correlation closer to the inter-user channel correlation of a LOS channel with independent angles. This observation again verifies OTA Finding $5$.

\subsubsection{Links to Systems without Inter-User Angle Correlation}
As mentioned in the spatial channel model~\eqref{equ:ray_mdl} setup, user channels without inter-user angle correlation can be viewed as specialized versions of the studied channel model by letting $\alpha_{x}=1$ and $\alpha_{y}=1$.
The link inspires us to compare this section's analysis with past results to highlight the impact of inter-user angle correlation.
This subsection begins with the independent AoD specification, which is followed by the LOS and Rayleigh fading cases.

\noindent
\emph{\sl Uniform Linear Array Specification}
We now adapt Propositions~\ref{prop:main_finite} and~\ref{prop:main_asy} to describe the inter-user channel correlation in the special case of ULA by letting $M_{y}=1$, $d=d_{x}$, and number of base-station antennas $M=M_{x}$.
Here, we consider the ULA as the past results on inter-user angle correlation~\cite{favorable_prop,wu2017favorable} adopted ULA.
For the above described ULA system, Proposition~\ref{prop:main_finite} gives the square of the inter-user channel correlation~\eqref{equ:alpha_def} $\Var \left[\frac{1}{M} \mathbf{h}_{k}^{H}\mathbf{h}_{k^{'}}\right]$ is
\begin{align}
 \sum_{l=1}^{L} \sum_{l^{'}=1}^{L}\left(\beta_{k,l} \beta_{k^{'},l^{'}}\right)^2 \eta\left(\frac{d \alpha_{l,l^{'}} }{\lambda} , M \right), \notag
\end{align}
where $\eta \left(c, M\right)=\sum_{m_{1}=0}^{M-1}\sum_{m_{2}=0}^{M-1} \sinc\left( 2\pi c \left(m_{1}-m_{2} \right) / \lambda\right)$, and $\alpha_{l,l^{'}}=\alpha_{x,l,l^{'}}$.
As $M\to\infty$, The square of inter-user channel correlation in the large-array regime is provided by Proposition~\ref{prop:main_asy} as
\begin{equation}
\sum_{l=1}^{L} \sum_{l^{'}=1}^{L}  \left(\beta_{k,l} \beta_{k^{'},l^{'}}\right)^2
\left(\bm{1}_{c = 0} + \bm{1}_{c \neq 0}  \frac{\lambda}{2 d \alpha_{l,l^{'}} M}\right).\label{equ:ULA_spec}
\end{equation}

\noindent
\emph{\sl LOS channel with independent inter-user angles:}
Past works labels users in LOS environments where $\sin\big(\theta_{k}\big) - \sin\big(\theta_{k}^{'}\big)$ follows i.i.d.\ uniform distribution over $\left[-1, 1\right]$ as
``Uniform Random LOS channel''~\cite{favorable_prop,wu2017favorable, matthaiou2018does, marzetta2016fundamentals}.
And~\cite{favorable_prop, matthaiou2018does}~\cite[Table 7.1]{marzetta2016fundamentals} showed that with half-wavelength spaced ULA, the square of the asymptotic inter-user channel correlation equals $\frac{1}{M}$ in the  ``Uniform Random LOS channel''.
Such result is immediate by substituting~\eqref{equ:ULA_spec} with $d=\frac{\lambda}{2}$, $\beta_{k,1}=\beta_{k^{'},1}=1$, $\alpha=1$,  and $L=1$.

\noindent
\emph{\sl The i.i.d.\ Rayleigh  fading channel:}
Finally, the past massive MIMO research~\cite{favorable_prop, marzetta2016fundamentals} demonstrated that in the i.i.d. Rayleigh fading channel, inter-user channel correlation satisfies that
\begin{equation}
  \Var \left[\frac{1}{M} \mathbf{h}_{k}^{H}\mathbf{h}_{k^{'}}\right] = \frac{1}{M}.~\label{equ:var_rayleigh}
\end{equation}
Recall that we can link the spatial channel model to the Rayleigh fading channel by $L\to \infty$.
When $L\to \infty$ and $\beta_{k,1}=\beta_{k^{'},1}=\frac{1}{\sqrt{L}}$, $\Var \left[\frac{1}{M}\mathbf{h}^{H}_{k}\mathbf{h}_{k^{'}}\right]$ in~\eqref{equ:ULA_spec} converges to $\frac{\lambda}{2 d \alpha_{l,l^{'}} M}$, which matches~\eqref{equ:var_rayleigh}.

We hence conclude that the past analytical results in both ``Uniform Random LOS'' channel and the i.i.d.\ Rayleigh fading channel without inter-user angle correlation are specialized versions of Proposition~\ref{prop:main_finite} and Proposition~\ref{prop:main_asy}.

\section{Numerical Evaluations}\label{sec:Numerical}
Section~\ref{sec:RayAnalysis} captures the inter-user channel correlation in massive MIMO systems with inter-user angle correlation in closed-form.
This section adopts numerical experiments to evaluate the accuracy of the theoretical analysis, Propositions~\ref{prop:main_finite} and Propositions~\ref{prop:main_asy}.
In addition, the simulation results confirm that inter-user angle correlation can result in the measured OTA inter-user channel correlation.
This section considers massive MIMO systems with UPA or ULA array configurations in the LOS and NLOS environments to match Section~\ref{sec:measurement} and Section~\ref{sec:RayAnalysis}.

In this section, we simulate the NLOS propagation environments via the keyhole channel~\cite{marzetta2016fundamentals} where the inter-user angles demonstrate pairwise correlation in both the $x$ and $y$ axes.
There are $L$ key-holes.
Similar to Section~\ref{subsec:theory_case_study},
the angles of user paths passing each key-hole have correlation of $\alpha$ in both the $x$ and $y$ axes.
And the angles of paths passing different key-holes are assumed to be independent with $\alpha=1$.
The LOS propagation environment is simulated as a special case of the keyhole channels by letting $L=1$. Unless otherwise specified, each simulation data point has a sample size of $20000$.
We note that the simulation in this section and the analysis in Section~\ref{sec:RayAnalysis} capture the inter-user channel correlation, which is not impacted by the choice of precoding algorithm.
\begin{figure*}[htbp]
\centering
\subfloat[UPA, LOS]{\label{fig:upa_alpha_los}
\includegraphics[width=0.42\textwidth]{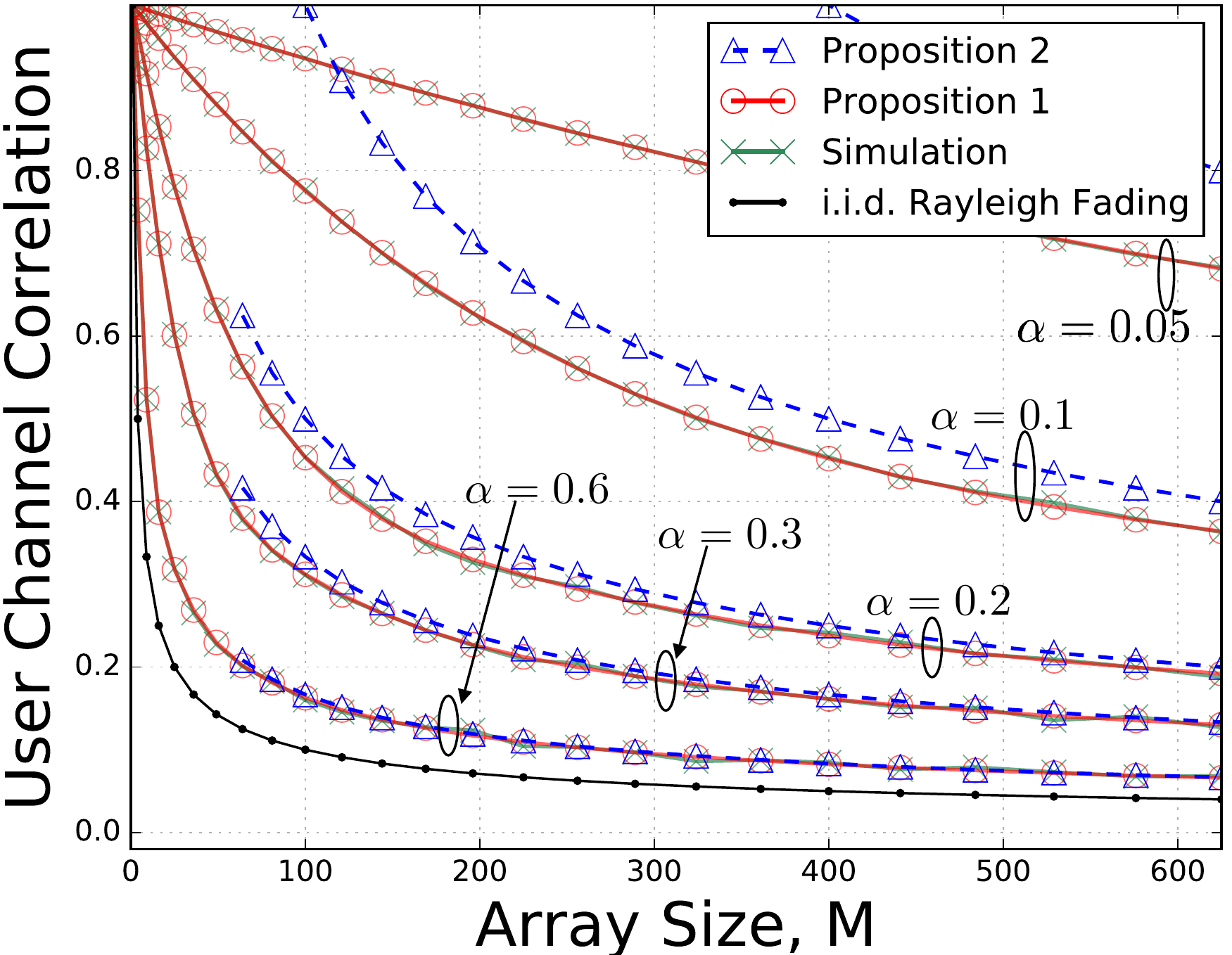}
}
  \subfloat[UPA, NLOS]{ \label{fig:upa_alpha_nlos}
\includegraphics[width=0.42\textwidth]{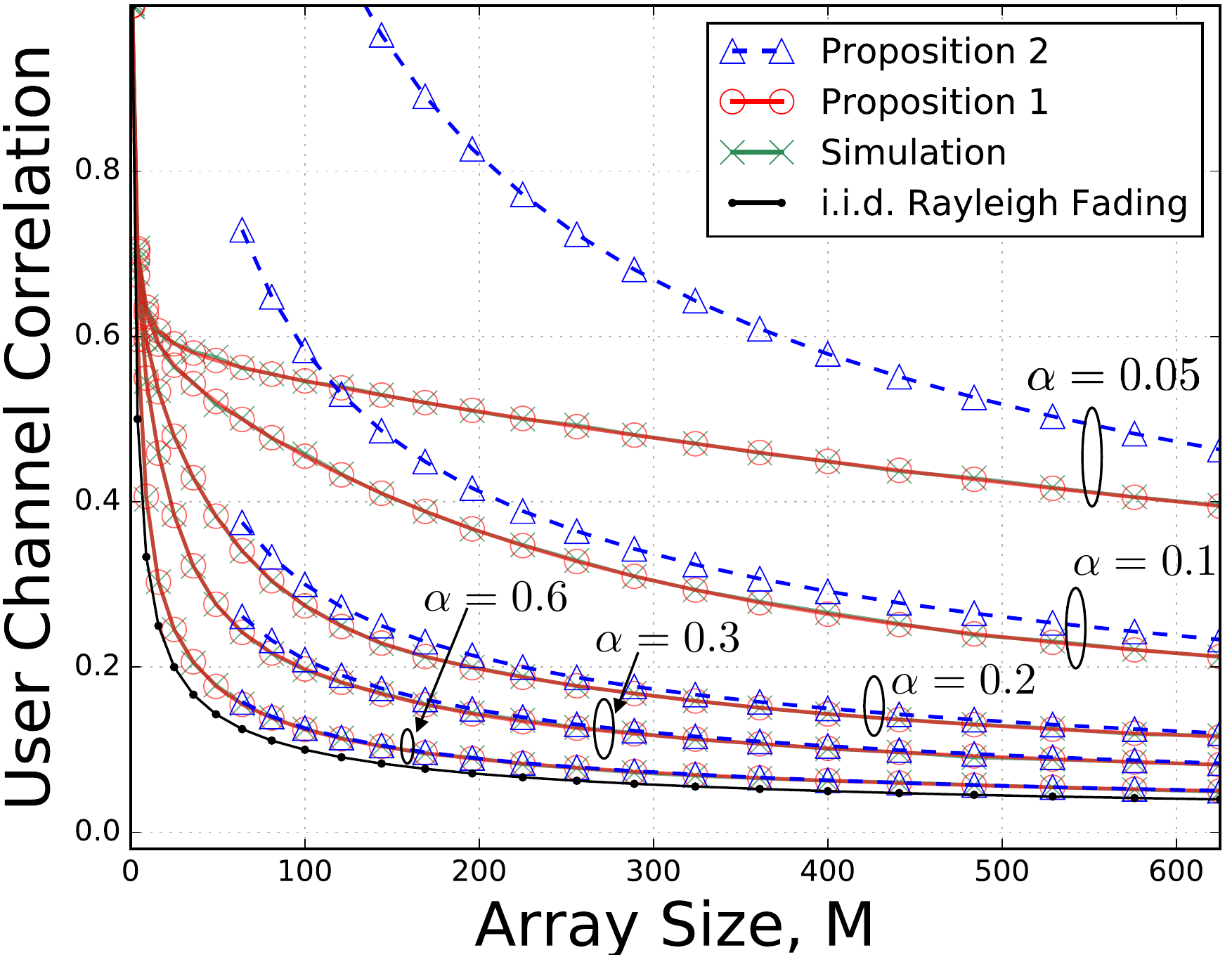}
  }
\caption{
Inter-user channel correlation for massive MIMO with inter-user angle correlation.
The base-station is equipped with a UPA in the LOS and NLOS environments.
Here, the antennas spacing is half wavelength, i.e. $d_x=d_y=0.5 \lambda $.
For each curve group, both $\alpha_x$ and $\alpha_y$ equal the labeled $\alpha$.
For the NLOS channel, $L=3$.
Fig~\ref{fig:upa_alpha} demonstrates that inter-user channel correlation reduces with $L$ and $\alpha$.
In other words, inter-user channel correlation reduces due to multi-path, and increases with inter-user angle correlation.
}
~\label{fig:upa_alpha}
\end{figure*}

Fig.~\ref{fig:upa_alpha} presents the inter-user channel correlation for UPA-based massive MIMO systems in both LOS and NLOS environments.
In Fig.~\ref{fig:upa_alpha_los} and Fig.~\ref{fig:upa_alpha_nlos}, the  inter-user angle correlation varies from $0.05$ to $0.6$.
The average inter-user channel correlation in the i.i.d.\  Rayleigh fading is presented in black.
The inter-user channel correlation reduces with the number of base-station antennas $M$ for all inter-user angle correlations.

In both LOS and NLOS environments, the simulations confirm that the inter-user channel correlation increases with inter-user angle correlation, which directly verifies the OTA Finding $3$ in Section~\ref{subsec:ota_findings}.
Combined with the observation that close-by user angles are more likely~\cite{liu2012cost} to have a high correlation, the simulation results verify the OTA Finding $1$ and Finding $2$.
By comparing the  inter-user channel correlation in Fig.~\ref{fig:upa_alpha_los} and Fig.~\ref{fig:upa_alpha_nlos}, the simulations confirm that the multi-path effect reduces the  inter-user channel correlation, i.e. OTA Finding $5$.

Recall that the i.i.d. Rayleigh fading channel can be viewed as a special case when $L=\infty$ and $\alpha=1$. Fig.~\ref{fig:upa_alpha_nlos} demonstrates that for low  inter-user angle correlation ($\alpha=0.6$) and a few spatial paths ($L=3$), a close to ideal inter-user channel correlation could be obtainable. However, in practical environments where inter-user angle correlation can be high, the inter-user channel correlation might be significant even when $M>600$.

Furthermore, Fig.~\ref{fig:upa_alpha} validates the accuracy of the finite-array and large-array theoretical analysis.
For all simulated $\alpha$ and $M$ in both environments, the finite-array analysis (red) exactly matches the numerical results (green).
And the large-array approximation (blue) in Proposition~\ref{prop:main_asy} becomes accurate when $M>200$ and $\alpha>0.1$.

Recalling that a ULA array is a special version of UPA with $M_y=1$, we now inspect the inter-user channel correlation of massive MIMO with a ULA base-station array.
Fig.~\ref{fig:ula_alpha} presents the inter-user channel correlation in both LOS and NLOS environments with varying levels of inter-user angle correlation.
Similar to the UPA counterparts in Fig.~\ref{fig:upa_alpha}, Fig.~\ref{fig:ula_alpha} shows that inter-user channel correlation increases with the  inter-user angle correlation.
Comparing Fig.~\ref{fig:ula_alpha_los} and Fig.~\ref{fig:ula_alpha_nlos} gives that multi-path effect reduces the inter-user channel correlation.
The overlap of Proposition~\ref{prop:main_finite} analysis (red) and simulated results (green) demonstrates the correctness of the finite-array analysis.
And we find that the large-array approximation in Proposition~\ref{prop:main_asy} becomes accurate for ULA systems when the number of antennas $M>150$.

Additionally, by comparing Fig.~\ref{fig:ula_alpha} to Fig.~\ref{fig:upa_alpha}, we find that ULA based systems demonstrate a lower inter-user channel correlation than the UPA counterparts.
The gap between UPA and ULA systems increases with inter-user angle correlation.
This observation suggests that the design of the base-station array configuration is essential for combating the inter-user angle correlation induced inter-user channel correlation.

\begin{figure*}[htbp]
\centering
\subfloat[ULA, LOS]{\label{fig:ula_alpha_los}
\includegraphics[width=0.42\textwidth]{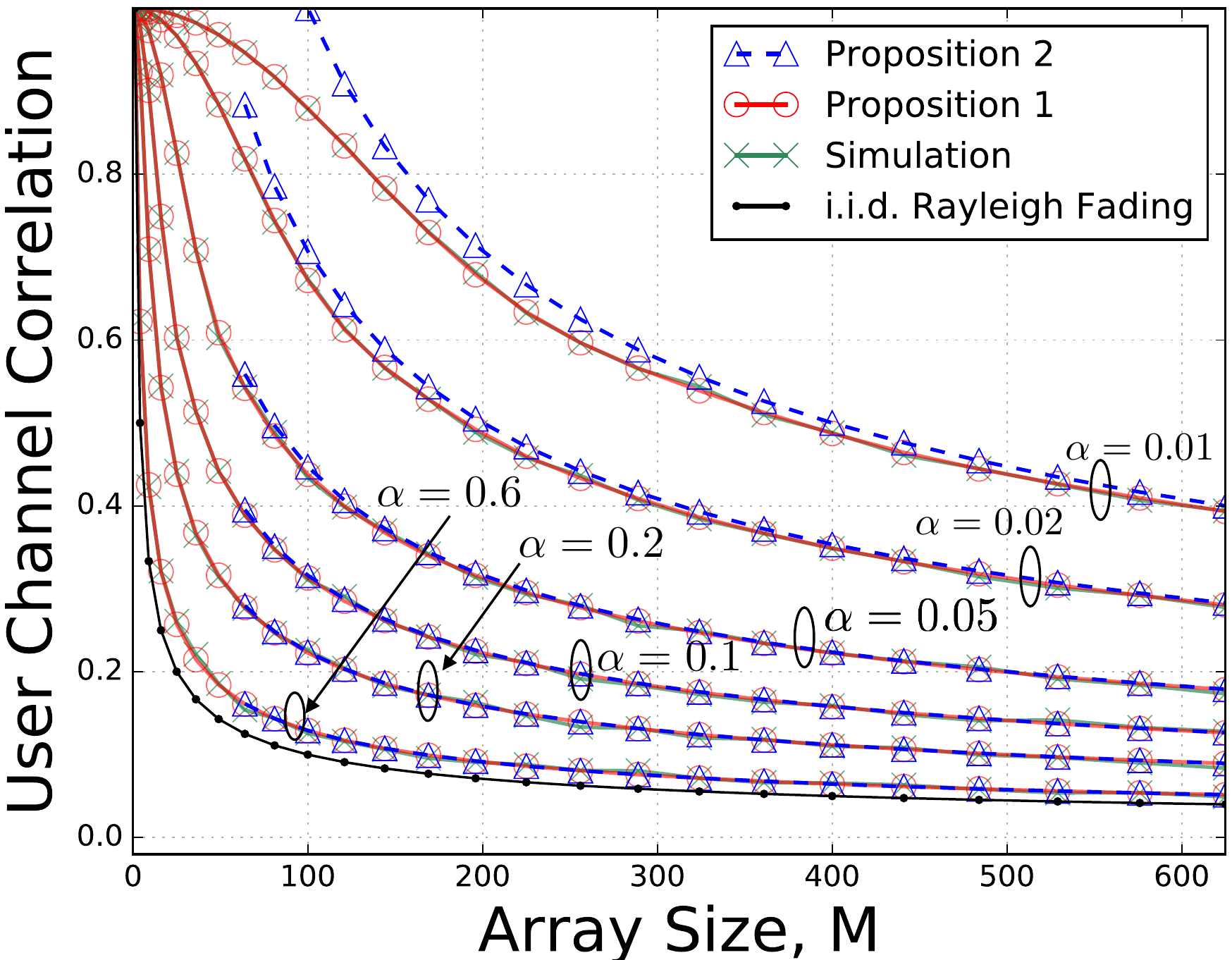}
}
  \subfloat[ULA, NLOS]{ \label{fig:ula_alpha_nlos}
\includegraphics[width=0.42\textwidth]{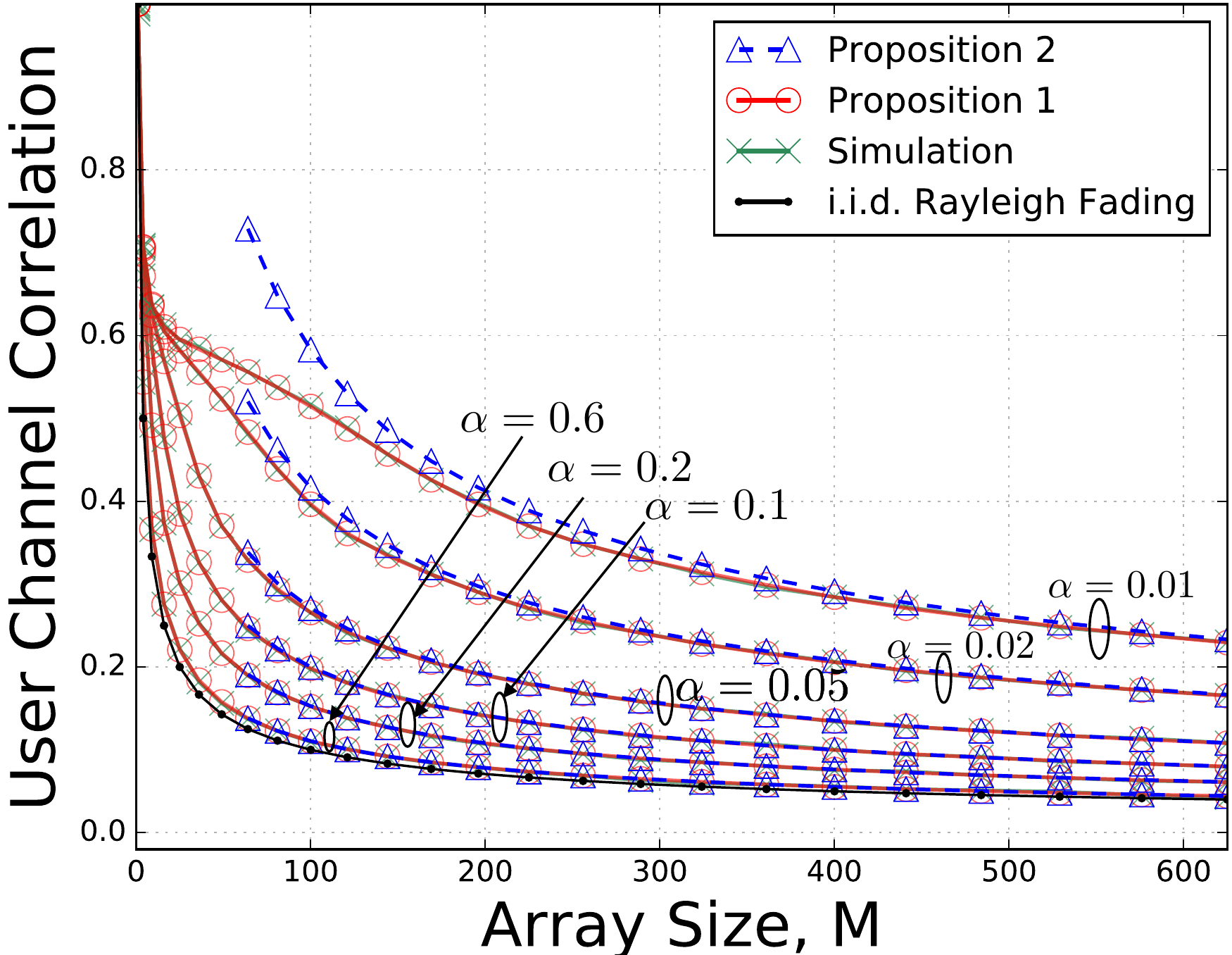}
  }
\caption{
Inter-user channel correlation for massive MIMO with the ULA array.
For both LOS and NLOS cases, the antenna spacing in $x$ and $y$ axes is half the wavelength.
For Fig.~\ref{fig:ula_alpha_nlos}, $L=3$.
Similar to the UPA counterpart in Fig.~\ref{fig:upa_alpha}, this figure demonstrates that multi-path and lower inter-user angle correlation reduces the inter-user channel correlation.
}
~\label{fig:ula_alpha}
\end{figure*}

Propositions~\ref{prop:main_finite} and~\ref{prop:main_asy} predicts that inter-user angle correlation increases the inter-user channel correlation as if
shrinking the base-station inter-antenna spacing while keeping the same $M$.
Fig.~\ref{fig:alpha_vs_delta} validates this prediction by presenting the inter-user channel correlation of users with different inter-user angle correlation and different base-station inter-antenna spacing.
The inter-user channel correlation in both the $x$ and the $y$ axes are $0.1$ or $0.2$.
Fig.~\ref{fig:alpha_vs_delta} considers antenna spacing being $0.25$ wavelength or $0.5$ wavelength.
For both UPA and ULA systems, we find that doubling the antenna spacing (from $0.25$ wavelength to $0.5$ wavelength) can exactly offset the effect of doubling the inter-user angle correlation ($\alpha=0.2$ to $\alpha=0.1$).
In other words, Fig.~\ref{fig:alpha_vs_delta} confirms that manipulating the base-station array spacing is a promising option for managing inter-user angle correlation induced inter-user channel correlation.

\begin{figure*}[htbp]
\centering
\subfloat[UPA, LOS]{\label{fig:alpha_vs_delta_upa}
\includegraphics[width=0.42\textwidth]{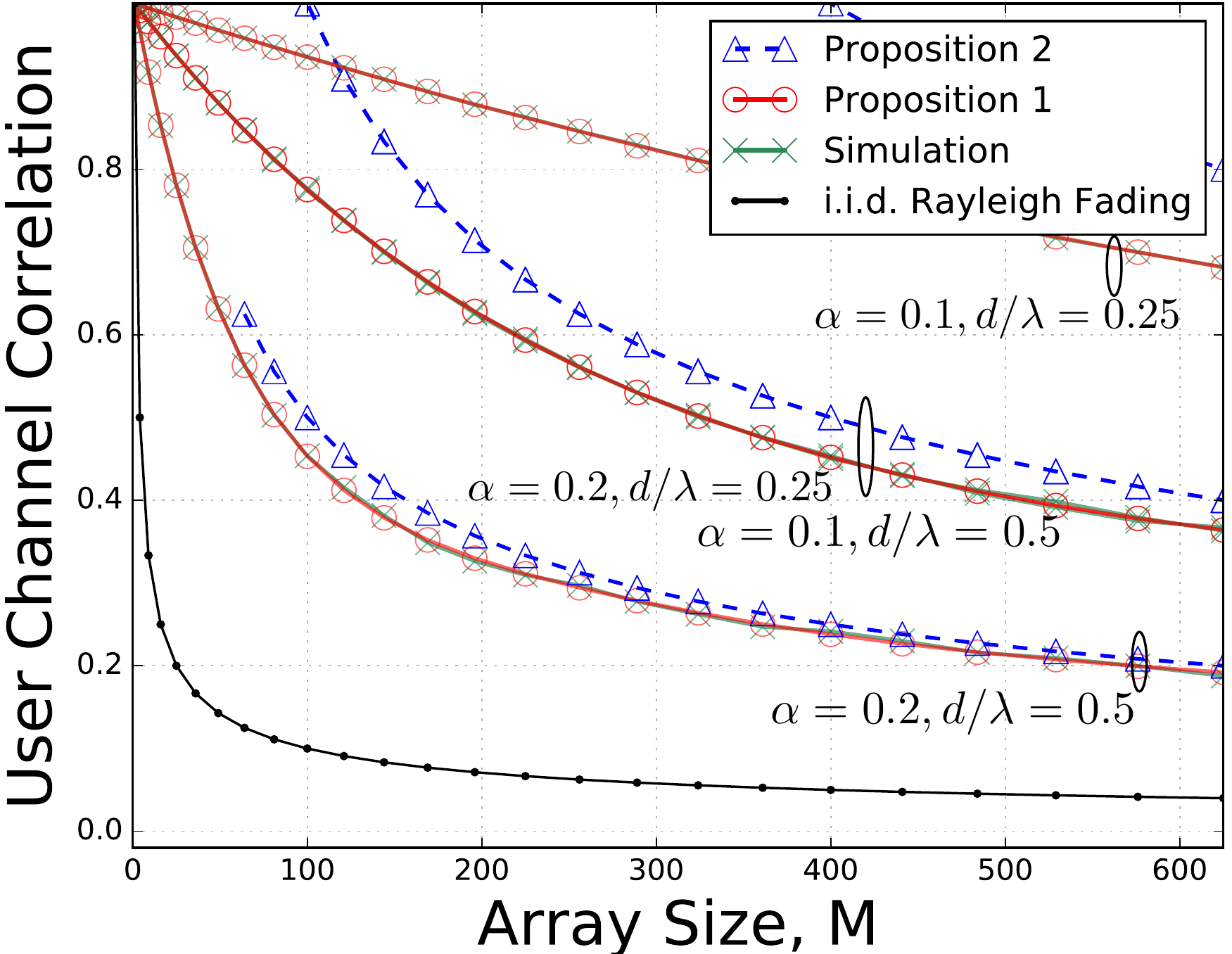}
}
  \subfloat[ULA, LOS]{ \label{fig:alpha_vs_delta_ula}
\includegraphics[width=0.42\textwidth]{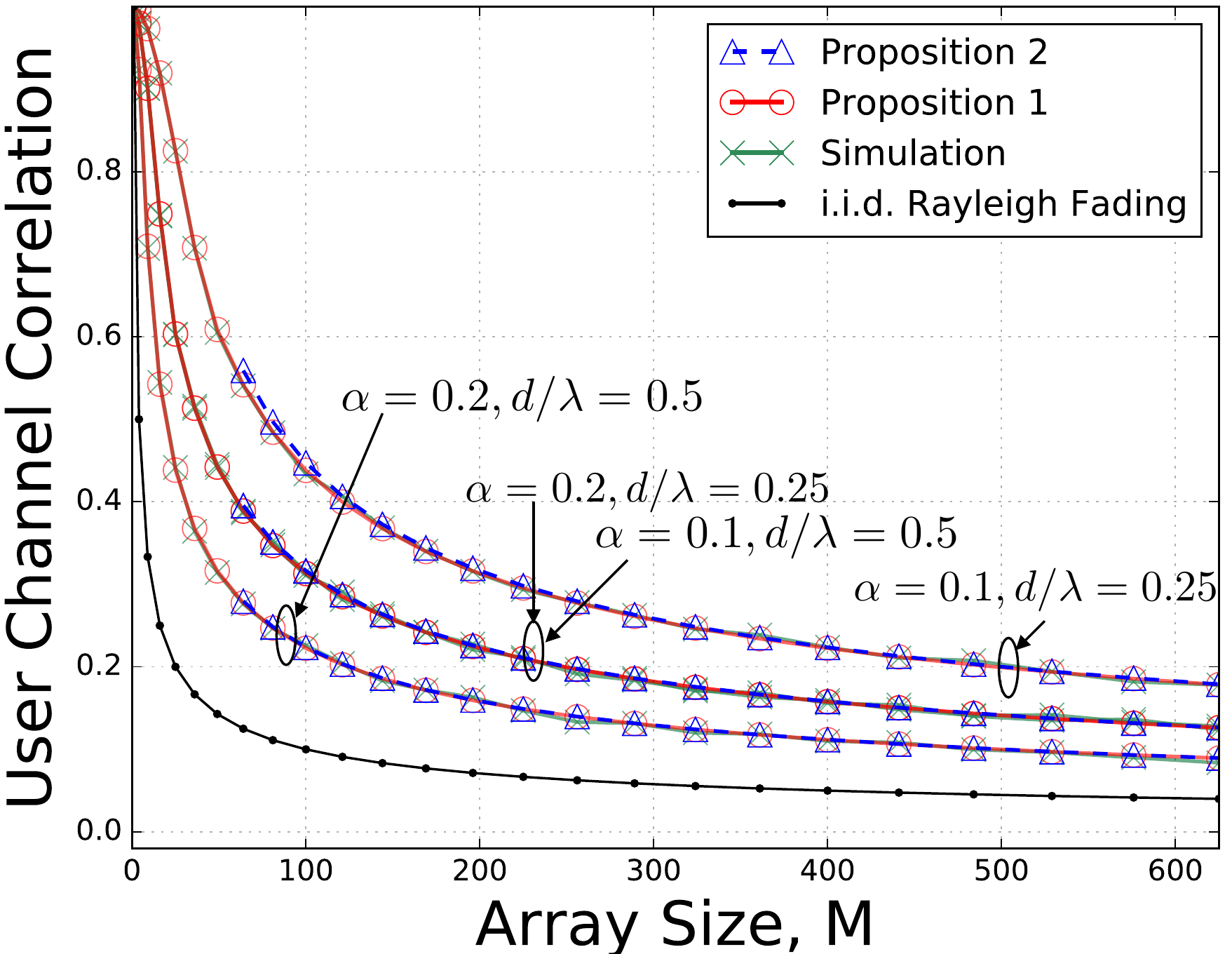}
  }
\caption{
Manipulating the base-station antenna spacing can offset the inter-user channel correlation increase caused by the inter-user angle correlation.
The labeled $d$ denotes the base-station array spacing in the $x$ and $y$ axes in this figure.
The inter-user angle correlation $\alpha_x$ and $\alpha_y$ equal the labeled $\alpha$.
}
~\label{fig:alpha_vs_delta}
\end{figure*}

We next use Fig.~\ref{fig:sub_car_var} to confirm the OTA finding that inter-user channel correlation is near-constant across the subcarriers.
Fig.~\ref{fig:sub_car_var} presents the inter-user channel correlation for UPA-based massive MIMO systems in both LOS and NLOS environments.
During the simulation, we emulate the OTA measurement setup in Section~\ref{sec:measurement} by using the same channel and base-station array setup.
The simulated channel has a central frequency of $2.437$ GHz with a $20$ MHz bandwidth.
There are $52$ data-carrying subcarriers.
The base-station adopts a UPA array with antenna spacing of $0.0635$ meters.
Fig.~\ref{fig:sub_car_var_los} and Fig.~\ref{fig:sub_car_var_nlos} presents the LOS and NLOS results, prospective.
The small error bars of the simulation curves (in green) confirm the OTA Finding $4$, and theoretical analysis in that inter-user channel correlation is near-constant across the subcarriers.
Finally, similar to the earlier simulation results, Fig.~\ref{fig:sub_car_var} also confirms the accuracy of the finite-array (red) and large-array (blue) theoretical analysis.

\begin{figure*}[htbp]
\centering
  \subfloat[UPA, LOS]{ \label{fig:sub_car_var_los}
\includegraphics[width=0.42\textwidth]{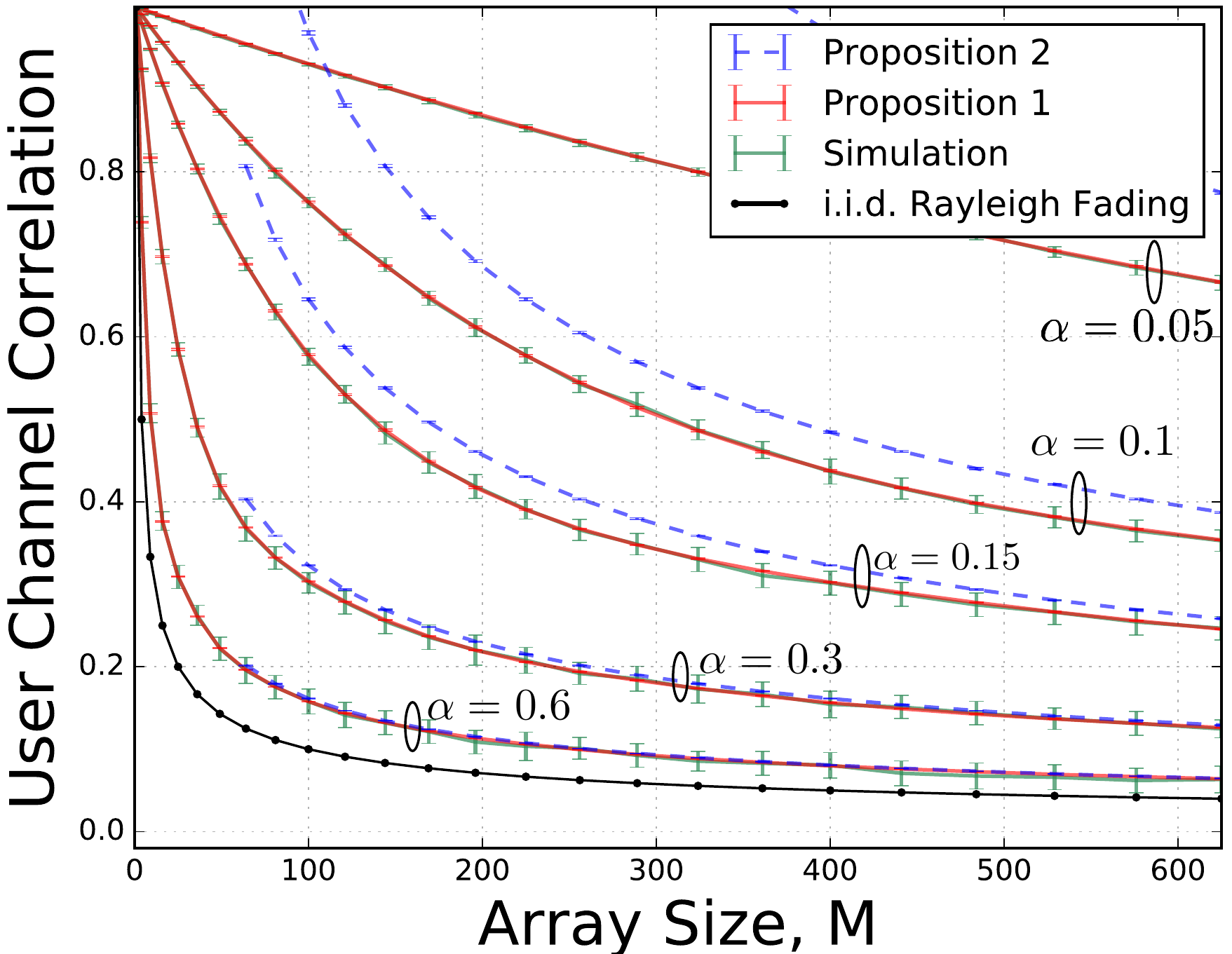}
}
\subfloat[UPA, NLOS]{\label{fig:sub_car_var_nlos}
\includegraphics[width=0.42\textwidth]{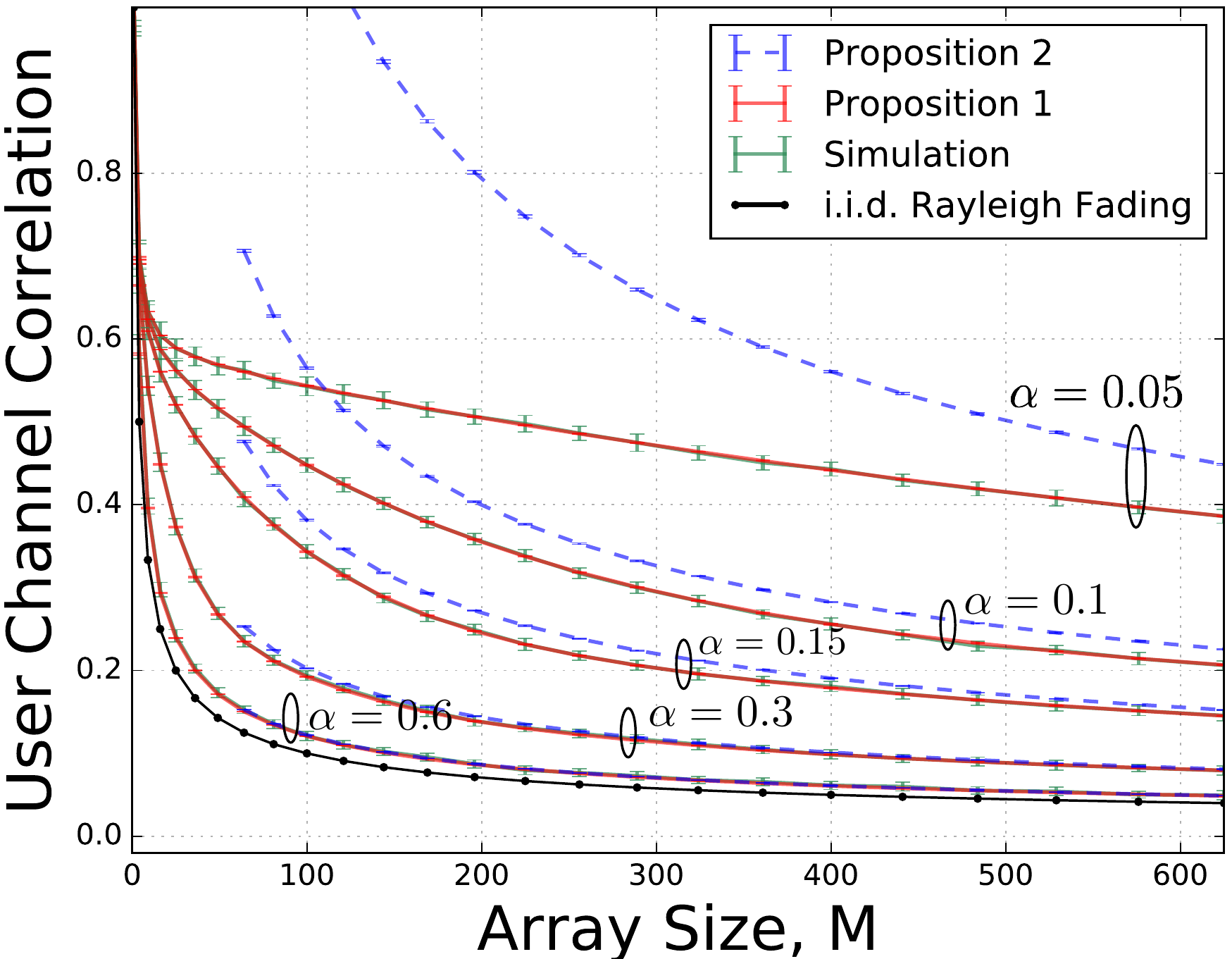}
}
\caption{
Inter-user channel correlation is near-constant across the subcarriers.
The error bar of each curve denotes the inter-user channel correlation among the 52 subcarriers across the  $20$ MHz bandwidth at the $2.4$ GHz.
The error bars for all simulation (green) curves confirm that the user channel standard deviation is less than $0.01$, which matches with OTA Finding $5$.
In Fig.~\ref{fig:sub_car_var_nlos}, $L=3$.
The base-station is equipped with a UPA and the labeled $\alpha$ represents both $\alpha_x$ and $\alpha_y$
The base-station antenna spacing satisfies that $d_{y}=0.0635$ m and $d_{y}=0.0635$ m.
For each curve, we simulate each subcarrier with $500$ samples.
}
~\label{fig:sub_car_var}
\end{figure*}

\section{Conclusion}\label{sec:Conclude}

This work studies the inter-user channel correlation in practical massive MIMO channels with inter-user angle correlation.
We measured and released a new massive MIMO channel dataset with over $11500$ unique channel vectors.
We use measured channels to examine the inter-user channel correlation in a real-world propagation environment.
The measurements show that adding more base-station antennas barely reduces the channel correlation between close-by users when the number of base-station antennas $M>36$.
Moreover, more than $30.56\%$ far-away user pairs have channel correlation twice higher than the correlation in the i.i.d. Rayleigh fading channel.
Via spatial signal processing, we find that users with high channel correlation can have proximity in the angle space.
We adopt spatial channel models to compute the inter-user channel correlation in closed-form with the impact of inter-user angle correlation.
The analysis proves that the inter-user channel correlation increases with inter-user angle correlation, and reduces with base-station array aperture.
With fixed base-station inter-antenna spacing, except for the corner cases where users always share at least one angle, the inter-user channel correlation converges to zero as $M$ increases to infinity.
However, the inter-user angle correlation and limited base-station form factor can result in high inter-user channel correlation in real-world systems.
We further validate the analysis with numerical experiments.
The presented results collectively demonstrate that inter-user angle correlation management, such as user selection, could be critical in practical environments.
Finally, base-station array configuration design is an attractive option for reducing inter-user channel correlation caused by inter-user angle correlation.

\begin{appendices}
\renewcommand{\thesectiondis}[2]{\Alph{section}:}
\section{Proof of User Channel Correlation, Proposition~\ref{prop:main_finite}}\label{appendix:main_finite}
We begin the proof by decomposing the inter-user channel product $ \mathbf{h}_{k}^{H}\mathbf{h}_{k^{'}} $ as
\begin{align}
&\left(\sum_{l=1}^{L}\beta_{k,l} e^{j\phi_{k,l}} \mathbf{a}\Big(\bm{\theta}_{k,l}\Big)\right)^{H}
                                   \left(\sum_{l=1}^{L} \beta_{k^{'},l^{'}}  e^{j\phi_{k^{'},l^{'}}} \mathbf{a}\left(\bm{\theta_{k^{'},l^{'}}}\right)\right)\notag \\
=& \sum_{l=1}^{L} \sum_{l^{'}=1}^{L} \beta_{k,l} \beta_{k^{'},l^{'}}  e^{j\left(\phi_{k^{'},l^{'}}-\phi_{k,l}\right)}  \mathbf{a}^{H}\Big(\bm{\theta_{k,l}}\Big)          \mathbf{a}\Big(\bm{\theta_{k^{'},l^{'}}}\Big) \notag
\end{align}
Under channel model~\eqref{equ:ray_mdl}, path phases $\phi_{k^{'},l^{'}}, \phi_{k,l}$ follow independent uniform distribution on $\left[0, 2\pi\right]$.
For any random variable $X$ that is independent of $\phi$, we have $\Exp \left[e^{j \phi} X\right]= 0$.
The linearity of the expectation operator then gives
\begin{equation}
\Exp \left[\frac{1}{M}\mathbf{h}_{k}^{H}\mathbf{h}_{k^{'}} \right] =  0, \label{equ:main_finite_mean}
\end{equation}
The variance~\eqref{equ:main_finite_var} can be then proved via some standard manipulations as
\begin{align}
\Var \left[\mathbf{h}_{k}^{H}\mathbf{h}_{k^{'}}\right] =& \Exp \left[\left|\mathbf{h}_{k}^{H}\mathbf{h}_{k^{'}}\right|^2\right] - \left|\Exp \left[\mathbf{h}_{k}^{H}\mathbf{h}_{k^{'}}\right]\right|^2 \notag \\
= & \Exp \left[\left|\mathbf{h}_{k}^{H}\mathbf{h}_{k^{'}}\right|^2\right], \notag
\end{align}
where the last step is via the proved mean~\eqref{equ:main_finite_mean}. Reusing the i.i.d. uniform path phase assumption, the variance term
$\Exp \left[\sum_{l=1}^{L} \sum_{l^{'}=1}^{L} \left(\beta_{k,l} \beta_{k^{'},l^{'}}\right)^2 \left| \mathbf{a}^{H}\Big(\bm{\theta_{k,l}}\Big) \mathbf{a}\Big(\bm{\theta_{k^{'},l^{'}}}\Big)\right|^2\right]$ equals
\begin{equation}
\sum_{l=1}^{L} \sum_{l^{'}=1}^{L} \left(\beta_{k,l} \beta_{k^{'},l^{'}}\right)^2 \Exp \left[  \left| \mathbf{a}^{H}\Big(\bm{\theta_{k,l}}\Big)          \mathbf{a}\Big(\bm{\theta_{k^{'},l^{'}}}\Big)\right|^2 \right]
\end{equation}
We complete this proof by computing each of the expectation terms $\Exp \left[ \left|
         \mathbf{a}^{H}\Big(\bm{\theta_{k,l}}\Big)
         \mathbf{a}\Big(\bm{\theta_{k^{'},l^{'}}}\Big)
         \right|^2 \right]$ as
\begin{align}
&\Var \left[\mathbf{a}^{H}\Big(\bm{\theta_{k,l}}\Big)
              \mathbf{a}\Big(\bm{\theta_{k^{'},l^{'}}}\Big)\right]
   + \Exp^{2} \left[
            \mathbf{a}^{H}\Big(\bm{\theta_{k,l}}\Big)
            \mathbf{a}\Big(\bm{\theta_{k^{'},l^{'}}}\Big)
         \right] \notag  \\
=& \eta\left(\frac{d_{x}\alpha_{x}}{\lambda} , M_{x}\right)  \eta \left(\frac{d_{y}\alpha_{y}}{\lambda}, M_{y}\right),
\end{align}
where the last step is by Lemma~\ref{appendix_prop:aa_correlation}.

\section{Asymptotic User Channel Correlation}\label{appendix:asy_proof}
Due to the symmetry between the $x$ axis and the $y$ axis, it is sufficient to prove that
\begin{equation}
\eta\left(c , M\right) \cong \tilde{\eta}\left(c, M\right)=\bm{1}_{c = 0} + \bm{1}_{c \neq 0}  \frac{1}{2 c M}. \label{equ:proof_asy}
\end{equation}
We prove by consider different values of $c$. When $c=0$, we use the definition of $\eta\left(\cdot\right)$  and have
\begin{equation}
\eta\left(0 , M\right) =
\frac{1}{M^2} \sum_{m_{1}=0}^{M-1}\sum_{m_{2}=0}^{M-1} \sinc\left( 0\right)= 1,\label{equ:proof_asy_zero}
\end{equation}
where the last step is by $\sinc\left(0\right)=1$.
When $c>0$, $\eta\left(c, M\right)$ can be computed with Lemma~\ref{append:lemma_2d_sinc_sum} as
\begin{align}
\lim_{M\to\infty} 2c M\eta\left(c , M\right)
= &  \lim_{M\to\infty} \frac{2c}{M} \sum_{m_{1}=0}^{M-1}\sum_{m_{2}=0}^{M-1}
                \sinc\left( 2\pi c \left(m_{1}-m_{2}\right)\right)  \notag \\
=& 1, \quad \text{for} \quad c \in \left(0, 1\right].  \label{equ:proof_asy_positive}
\end{align}
Combining~\eqref{equ:proof_asy_zero} and~\eqref{equ:proof_asy_positive} gives~\eqref{equ:proof_asy}, which completes the proof.

\section{Useful Properties Of Spatial Correlated Paths}
\begin{lemma}~\label{appendix_prop:aa_correlation}
For two array response vectors with angles that satisfy $\sin \theta_{x,1}  -\sin \theta_{x,2} \sim U\left[-\alpha_{x}, \alpha_{x}\right]$, $\sin\theta_{y,1} -\sin \sin\theta_{y,2} \sim U\left[-\alpha_{y},\alpha_{y}\right]$ and is independent of $\sin \theta_{x,1}  -\sin \theta_{x,2}$,
and $\alpha_{x}, \alpha_{y} \in [0, 1]$, the mean and variance of their response vector follows
\begin{align}
\Exp \left[\mathbf{a}^{H}\left(\bm{\theta_{1}}\right) \mathbf{a}\left(\bm{\theta_{2}}\right) \right] = &
\mu\left(\frac{\alpha_{x} d_{x}}{\lambda}, M_{x}\right)
\mu\left(\frac{\alpha_{y}d_{y}}{\lambda}, M_{y}\right), \label{prop_equ:aa_mean}\\
\Var \left[\mathbf{a}^{H}\left(\bm{\theta_{1}}\right) \mathbf{a}\left(\bm{\theta_{2}}\right)\right] =&
\eta\left(\frac{\alpha_{x} d_{x}}{\lambda} , M_{x}\right)  \eta \left(\frac{\alpha_{y} d_{y}}{\lambda}, M_{y}\right)
-  \notag \\
& \left[
\mu\left(\frac{\alpha_{x} d_{x}}{\lambda}, M_{x}\right)
\mu\left(\frac{\alpha_{y}d_{y}}{\lambda}, M_{y}\right)
\right]^2, \label{equ:append_c_var}
\end{align}
where $\mu\left(c, M\right) = \sum_{m=0}^{M-1} \sinc\left( 2\pi m c \right)$,
and $\eta \left(c, M\right)= \sum_{m_{1}=0}^{M-1}\sum_{m_{2}=0}^{M-1} \sinc\left( 2\pi c \left(m_{1}-m_{2} \right) \right).$
\end{lemma}

\begin{proof}
 \begin{figure*}
\centering
\begin{minipage}{1.0\textwidth}
\begin{align}
\Exp \left[\mathbf{a}^{H}\left(\bm{\theta_{1}}\right) \mathbf{a}\left(\bm{\theta}_{2}\right) \right] =&
\sum_{m_{x}=0}^{M_{x}-1}\sum_{m_{y}=0}^{M_{y}-1} \Exp \left[
e^{-j  2\pi \left(m_{x} \sin \theta_{x,1} \frac{d_{x}}{\lambda} + m_{y} \sin \theta_{y,1} \frac{d_{y}}{\lambda}\right)}
e^{j 2\pi \left(m_{x} \sin \theta_{x,2} \frac{d_{x}}{\lambda} + m_{y} \sin \theta_{y,2} \frac{d_{y}}{\lambda}\right)}
\right]  \notag\\
=& \sum_{m_{x}=0}^{M_{x}-1}\sum_{m_{y}=0}^{M_{y}-1} \int_{-\alpha_{y}}^{\alpha_{y}} \frac{1}{2\alpha_{y}}
\int_{-\alpha_{x}}^{\alpha_{x}} \frac{1}{2 \alpha_{x}}
e^{j 2\pi \left(m_{x} u_{x} \frac{d_{x}}{\lambda} + m_{y} u_{y} \frac{d_{y}}{\lambda}\right)}
 d u_{x} d u_{y} \notag\\
=&  \sum_{m_{x}=0}^{M_{x}-1}\sum_{m_{y}=0}^{M_{y}-1} \sinc\left( 2\pi m_{x} \frac{\alpha_{x} d_{x}}{\lambda} \right)
\int_{-\alpha_{y}}^{\alpha_{y}} \frac{1}{2\alpha_{y}}
e^{j 2\pi  m_{y} u_{y} \frac{d_{y}}{\lambda}}   d u_{y}  \notag \\
=& \sum_{m_{x}=0}^{M_{x}-1}  \sinc\left( 2\pi m_{x} \frac{\alpha_{x} d_{x}}{\lambda} \right)
    \sum_{m_{y}=0}^{M_{y}-1} \sinc\left( 2\pi m_{y}  \frac{\alpha_{y}d_{y}}{\lambda} \right). \label{equ:temp_aa_sinc_proof}
\end{align}
\hrule
\end{minipage}
\end{figure*}

Define $u_{x}  = \sin(\theta_{x,2}) - \sin(\theta_{x,1})$, and $u_{y}  = \sin(\theta_{y,2}) - \sin(\theta_{y,1})$.
We compute the first moment~\eqref{prop_equ:aa_mean} in~\eqref{equ:temp_aa_sinc_proof} by using the definition of the UPA array response vector~\eqref{equ:def_a_theta_UPA}.

Similarly, we compute the second moment~\eqref{equ:append_c_var} by~\eqref{equ:def_a_theta_UPA} as

\begin{align}
\Exp\left[\left|\mathbf{a}^{H}\left(\bm{\theta_{1}}\right) \mathbf{a}\left(\bm{\theta_{2}}\right)\right|^2\right] - \left|\Exp\left[\mathbf{a}^{H}\left(\bm{\theta_{1}}\right) \mathbf{a}\left(\bm{\theta_{2}}\right)\right]\right|^2.\label{equ:nov_1_temp2}
\end{align}
Recall the first moment~\eqref{prop_equ:aa_mean} is proved by~\eqref{equ:temp_aa_sinc_proof}, we complete the proof by computing $\Exp\left[\left|\mathbf{a}^{H}\left(\bm{\theta_{1}}\right) \mathbf{a}\left(\bm{\theta_{2}}\right)\right|^2\right]$, which equals
\begin{equation}
\Exp\left[\left|
\sum_{m_{x}=0}^{M_{x}-1}\sum_{m_{y}=0}^{M_{y}-1}
e^{
j  2\pi
\left(
\left(m_{x_{1}} -m_{x_{2}} \right) u_{x} \frac{d_{x}}{\lambda} +
\left(m_{y_{1}} -m_{y_{2}} \right) u_{y} \frac{d_{y}}{\lambda}
\right)
}
\right|^2
\right],\notag
\end{equation}
which is simplified by using independence of $u_{x}$ and $u_{y}$ as
\begin{align}
\Exp\Bigg[ \sum_{m_{x_{1}}=0}^{M_{x}-1} &\sum_{m_{x_{2}}=0}^{M_{x}-1}
e^{j  2\pi  \left(m_{x_1} - m_{x_2}\right) u_{x} \frac{d_{x}}{\lambda} }  \times \notag
\\
& \sum_{m_{y_{1}}=0}^{M_{y}-1}\sum_{m_{y_{2}}=0}^{M_{y}-1}
e^{j  2\pi \left(m_{y_1} - m_{y_2}\right) u_{y} \frac{d_{y}}{\lambda}}
\Bigg]
\label{equ:nov_1_temp}
\end{align}
By the linearity of expectation operator and following steps similar to~\eqref{equ:temp_aa_sinc_proof}, the proof is completed by~\eqref{equ:nov_1_temp} and~\eqref{equ:nov_1_temp2}.
\end{proof}

\section{Some Useful Lemmas}~\label{append:user_cor}
\begin{lemma} ~\label{append:lemma_sinc_sum}
  For any $0 < \gamma \leq \pi$, we have
  \begin{equation}
    \sum_{m=0}^{\infty} \sinc\left(\gamma m\right) = \sum_{m=0}^{\infty} \frac{\sin\left(\gamma m\right)}{\gamma m} = \frac{\pi + \gamma}{2 \gamma}.
  \end{equation}
\end{lemma}
\begin{proof}
The proof is immediate via using the Poisson summation formula. The key step is to observe that the rectangular window function's Discrete-Time Fourier Transform (DTFT) is the $\sinc$ function. Reference~\cite[5.4]{prudnikov1986integrals} provides detailed proof.
\end{proof}

\begin{lemma} ~\label{append:lemma_2d_sinc_sum}
  For any positive real number $\gamma > 0$, we have
  \begin{equation}
  \lim_{M \to \infty} \frac{1}{M} \sum_{m_{1}=0}^{M} \sum_{m_{2}=0}^{M} \sinc\left(\gamma \left(m_{1} - m_{2}\right)\right)
  = \frac{\pi}{\gamma}. \label{equ:append_a_2}
  \end{equation}
\end{lemma}
\begin{proof}
  When $\gamma=\pi$, the left hand of~\eqref{equ:append_a_2} equals $1$. DTFT's scaling property and Poisson summation formula complete the proof.
\end{proof}

\end{appendices}

\section*{Acknowledgment}
We will like to thank Dr. Xing Zhang for his help during the channel data measurement.
We also want to thank Dr. Clayton Shepard for building and supporting the Argos massive MIMO platform~\cite{shepard2012argos}.

\bibliography{libs/IEEEabrv,libs/latencyReTx}

\begin{IEEEbiography}
    [{\includegraphics[width=1in,height=1.25in,clip,keepaspectratio]{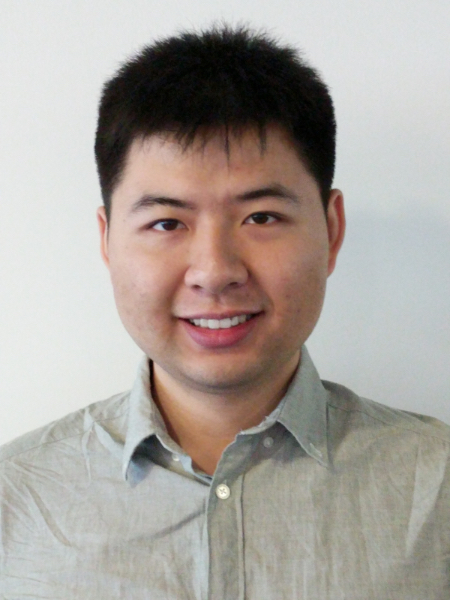}}]{Xu Du}
 received the B.E. degree in information technology from Zhejiang University, China, in 2013. He also received the M.S., Ph.D. degree in electrical and computer engineering from Rice University, Houston, TX, USA, in 2015, and 2019, respectively. He is currently a research scientist with Facebook Inc., Menlo Park, CA, USA. His research interest includes computer system optimization, information theory, and algorithm design.
\end{IEEEbiography}

\begin{IEEEbiography}
    [{\includegraphics[width=1in,height=1.25in,clip,keepaspectratio]{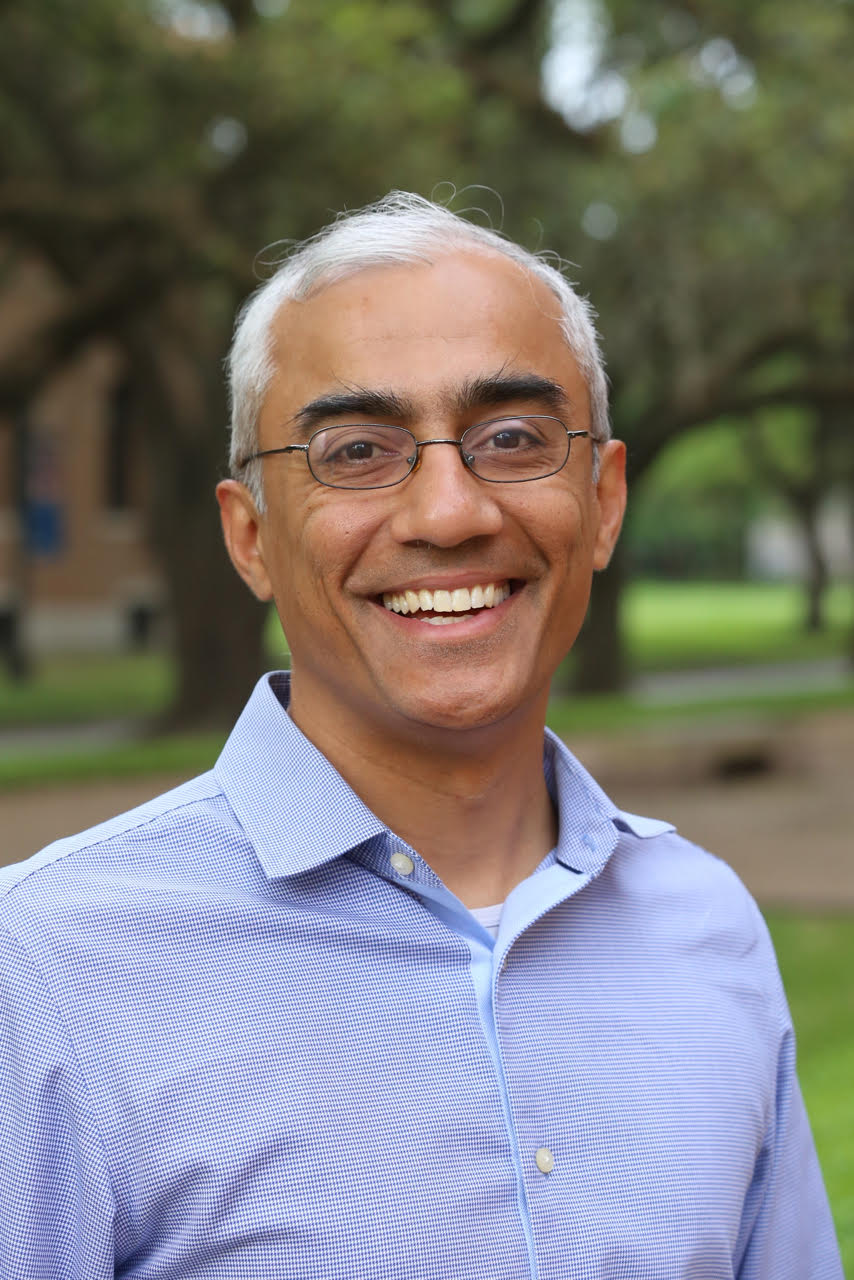}}]{Ashutosh Sabharwal}
(IEEE Fellow 2014) research interests are in wireless theory, protocols and open-source platforms. He is the founder of WARP project (warp.rice.edu), an open-source project which was used at more than 150 research groups worldwide, and have been used by more than 500 research articles. He received 2017 Jack Neubauer Memorial Award, 2018 Advances in Communications Award, 2019 and 2021 ACM Sigmobile Test-of-time Awards, and 2019 Mobicom Best Community Contribution Paper Award. He is a Fellow of the  National Academy of Inventors.
\end{IEEEbiography}

\clearpage

\end{document}